\documentclass[10pt]{article}

\usepackage[preprint]{neurips_2026}

\usepackage{amsmath,amssymb,amsthm}
\usepackage{enumitem} 

\usepackage{booktabs}
\usepackage{tabularx}
\usepackage{array}
\newcolumntype{Y}{>{\raggedright\arraybackslash}X}

\usepackage{hyperref}
\usepackage{algorithm}
\usepackage{algorithmic}
\usepackage{tikz}
\usetikzlibrary{calc}
\usetikzlibrary{arrows.meta, positioning, fit}
\usepackage{comment}
\usepackage{booktabs}

\usepackage{times}
\usepackage{color}
\usepackage{soul}

\usepackage{pgfplots}
\usepackage{subcaption}
\usepackage{pgfplots}
\pgfplotsset{compat=1.18}
\usepgfplotslibrary{groupplots,colormaps}

\usetikzlibrary{arrows.meta,positioning,calc,decorations.pathreplacing}

\definecolor{darkblue}{rgb}{0.0,0.0,0.2}
\definecolor{darkgreen}{rgb}{0.0,0.3,0.0}
\hypersetup{colorlinks,breaklinks,
	linkcolor=darkblue,urlcolor=darkblue,
	anchorcolor=darkblue,citecolor=darkblue}
\usepackage[colorinlistoftodos,textsize=tiny]{todonotes} 

\newcommand{\Comments}{1}
\newcommand{\mynote}[2]{\ifnum\Comments=1\textcolor{#1}{#2}\fi}
\newcommand{\mytodo}[2]{\ifnum\Comments=1	\todo[linecolor=#1!80!black,backgroundcolor=#1,bordercolor=#1!80!black]{#2}\fi}

\newcommand{\rickt}[1]{\mytodo{red!20!white}{RN: #1}}

\ifnum\Comments=1              
\fi

\usepackage{thm-restate}
\usepackage{thmtools}

\newtheorem{theorem}{Theorem}
\newtheorem{assumption}{Assumption}

\newtheorem{corollary}{Corollary}
\newtheorem{lemma}{Lemma}

\newtheorem{definition}{Definition}
\newtheorem{axiom}{Axiom}

\newcommand{\fee}{\mathrm{fee}}

\newcommand{\Loss}{\mathrm{Loss}}
\newcommand{\Pay}{\mathrm{Pay}}
\newcommand{\mix}{\mathrm{mix}}

\newcommand{\dom}{\mathrm{dom}}
\newcommand{\relint}{\mathrm{relint}}
\newcommand{\Var}{\mathrm{Var}}
\newcommand{\Cov}{\mathrm{Cov}}

\newcommand{\reals}{\mathbb{R}}

\newcommand{\R}{\mathbb{R}}
\newcommand{\E}{\mathbb{E}}

\newcommand{\op}{\mathrm{op}}

\newcommand{\slip}{\mathrm{slip}}

\newcommand{\dyn}{\mathrm{dyn}}

\newcommand{\hyb}{\mathrm{hyb}}
\newcommand{\surr}{\mathrm{surr}}

\title{
Adaptive Liquidity in Prediction Markets via Online Learning
}
\author{
Enrique Nueve\\
Department of Computer Science\\
University of Colorado Boulder\\
\texttt{enrique.nueveiv@colorado.edu} \\
\And 
Bao Nguyen\\
Department of Applied Mathematics \\
University of Colorado Boulder\\
\texttt{bao.nguyen-2@colorado.edu} \\
\AND
Rafael Frongillo \\
Department of Computer Science\\
University of Colorado Boulder \\
\texttt{raf@colorado.edu}
\And
Bo Waggoner \\
Department of Computer Science\\
University of Colorado Boulder\\
\texttt{bwag@colorado.edu} \\
}

\begin{document}
\maketitle


\begin{abstract}
Prediction markets rely on liquidity to convert trades into informative prices, yet existing mechanisms fix liquidity ex ante. This restriction enforces a static trade-off between price responsiveness and worst-case loss despite inherently nonstationary trading conditions. We propose a fundamentally different approach that treats liquidity selection itself as an online learning problem. Our mechanism mixes a family of cost-function markets via learnable weights, yielding a single adaptive market that preserves no-arbitrage, bounded worst-case loss, expressiveness, and positive upside. 
We introduce a hybrid structural risk signal—a per-round objective that
quantifies the trade-off between price impact and inventory risk—and show
that standard online learning algorithms achieve switching-regret guarantees
relative to the best sequence of liquidity regimes in hindsight.
Simulations demonstrate that the mechanism adaptively shifts liquidity across regimes in response to both order flow and inventory dynamics. Our results establish a principled framework for adaptive liquidity, connecting prediction market design with online learning.
\end{abstract}

\section{Introduction and Related Work}

Prediction markets aggregate information by allowing traders to buy and
sell securities whose payoffs depend on future events
\citep{hanson2003combinatorial}. Automated market makers implement this
mechanism via a convex cost function $C$, which determines both prices
$p(q)=\nabla C(q)$ and payments $C(q+r)-C(q)$. This convex structure
ensures path independence, no-arbitrage, and bounded worst-case loss. A
central design choice is liquidity, governed by the curvature of $C$:
higher curvature leads to larger price impact and higher trading costs,
while flatter curvature improves execution but increases the market
maker’s exposure to inventory risk. In this sense, liquidity captures how
sensitive prices are to trades.

In practice, liquidity demand is not static. Trading intensity,
information arrival, and participation evolve over time, so a single
liquidity level is unlikely to be appropriate throughout the lifetime of
a market. Classical mechanisms such as the Log Market Scoring Rule (LMSR)
fix liquidity ex ante, imposing a constant trade-off between price
responsiveness and risk. Prior work has explored limited forms of dynamic
liquidity. In particular, \cite{abernethy2014general} introduces a
volume-based parameterization in which liquidity increases with
cumulative trading volume. While this captures the growth of market
depth, the resulting liquidity is monotone and cannot respond to regime
changes such as shifts in order flow or inventory exposure. This
limitation motivates mechanisms that adapt liquidity online in a
state-dependent manner.

We propose an adaptive market maker that combines a family of cost
functions $\{C_k\}_{k=1}^M$, each representing a different liquidity
regime. At each round, the market uses a convex mixture of these costs
with time-varying weights, producing a single pricing rule. This
construction preserves the standard guarantees of cost-function markets,
while allowing liquidity to shift across regimes over time.

To guide this adaptation, we introduce a hybrid signal that combines two
economically meaningful quantities: slippage and liability. Slippage
captures the realized price impact of trades under a given liquidity
regime, and we center this term across experts to measure relative
execution quality in each round. Liability measures worst-case inventory
exposure and is left uncentered to retain its interpretation as an
absolute risk. Together, these terms encode a structural trade-off:
high-liquidity regimes reduce price impact, while low-liquidity regimes
limit risk. We use this signal within an online learning framework to
update the mixture weights over time.

We analyze the resulting mechanism and signal through regret minimization by
relating the learning objective to the behavior of the implemented
market. Our results show that the adaptive market tracks the performance
of the best sequence of $J$ liquidity regimes in hindsight relative to our signal, achieving
$O(\sqrt{T\log T})$ regret under mild burn-in assumptions. Simulations
illustrate how the mechanism responds to changing market conditions,
increasing liquidity in periods of sustained trading activity and
reducing it as inventory risk accumulates.

Our main contributions are as follows. We introduce an adaptive
cost-function market that combines liquidity regimes via learnable
weights while preserving no-arbitrage, bounded loss, expressiveness, and positive upside.
We propose a hybrid structural signal, combining centered slippage and
non-centered liability, that captures the trade-off between execution
quality and risk, and we provide regret guarantees for adaptive
liquidity selection. We further show that the realized market behavior
aligns with the learning objective up to controlled distortions arising
from pricing mismatch and weight updates. Finally, we demonstrate
state-dependent liquidity adaptation across regimes in simulation.

\paragraph{Related work.}
A large body of work studies liquidity in automated prediction markets,
including \citep{othman2011liquidity, othman2012profit, othman2013practical, li2013axiomatic},
which analyze its impact on pricing, efficiency, and profitability.
However, these works focus on fixed mechanisms and do not consider
adaptive approaches that combine multiple market makers. In particular,
they do not integrate efficient cost-function markets
\citep{abernethy2013efficient} with online learning over a family of
liquidity regimes, which is the focus of our approach.

Our mixture construction builds on the log-sum-exp (softmax) operator,
which underlies exponential weighting in online learning
\citep{freund1997decision, cesa2006prediction}, the log-partition
function in statistical models \citep{wainwright2008graphical, murphy2012machine},
and convex smoothing and geometric programming formulations
\citep{boyd2004convex, boyd2007tutorial}. We leverage these properties to
combine cost functions into a single convex potential with a natural
interpretation as expert aggregation.

Finally, our signal design is informed by the connection between
constant-function market makers and prediction markets
\citep{frongillo2024axiomatic}, as well as the decentralized finance
literature, where liquidity is characterized through slippage—the
sensitivity of prices to trades \citep{engel2022presentation}. This perspective motivates our use of slippage as a
core component of the learning signal.


 
\section{Background: Cost-Function Markets and Adaptive Liquidity}
\label{sec:background}

\subsection{Cost-Function Market Makers}

Let $\mathcal O := \{o_i\}_{i=1}^{d}$ be a finite outcome space, and let
$\rho(o)\in\mathbb{R}^d$ denote the payoff vector realized upon outcome
$o \in \mathcal O$. A trader holding inventory $q\in\mathbb{R}^d$
receives payoff $\rho(o)\cdot q$.
A cost-function market maker is specified by a convex, differentiable
function $C:\mathbb{R}^d\to\mathbb{R}$. When a trader submits a trade
$r\in\mathbb{R}^d$, the market maker charges
\[
\mathrm{Pay}(q,r)=C(q+r)-C(q),
\]
and updates the market state to $q \leftarrow q+r$. For any initial state
$q_0\in\mathbb{R}^d$ and sequence of trades $(r_t)_{t=1}^T$ with
$q_{t}=q_{t-1}+r_t$, payments telescope:
\[
\sum_{t=1}^T \mathrm{Pay}(q_{t-1},r_t)
=
C(q_T)-C(q_0).
\]

Following \cite{abernethy2013efficient}, desirable properties of a
cost-function market are captured by the following axioms.

\begin{axiom}[Existence of instantaneous prices]
\label{axiom:price}
$C$ is continuous and differentiable on $\mathbb{R}^d$.
\end{axiom}

\begin{axiom}[Information incorporation]
\label{axiom:info}
For all $q,r\in\mathbb{R}^d$, $\mathrm{Pay}(q+r,r)\ge \mathrm{Pay}(q,r)$. 
\end{axiom}

\begin{axiom}[No arbitrage]
\label{axiom:noarb}
For any trade sequence $s=(r_1,\dots,r_T)\in(\mathbb{R}^d)^*$ and initial
state $q_0\in\mathbb{R}^d$,
\[
\min_{o\in\mathcal O}
\Big\langle \rho(o),\, \sum_{i=1}^T r_i \Big\rangle
\le
\sum_{i=1}^T \mathrm{Pay}(q_{i-1},r_i).
\]
\end{axiom}

\begin{axiom}[Expressiveness]
\label{axiom:express}
For any $p\in\Delta_d := \{ p \in \reals_{+}^{d} \mid \|p\|_1 =1 \}$, we write $x^p := \mathbb{E}_{o\sim p}[\rho (o)]$. 
Then for any $p\in\Delta_d$ and $\epsilon>0$, there exists $q\in\mathbb{R}^d$
for which $\|\nabla C(q)-x^p\|<\epsilon$.
\end{axiom}
A general construction of cost functions satisfying these axioms is given
by the convex conjugate representation.

\begin{theorem}[\cite{abernethy2013efficient}, Thm.~4.2]
\label{thm:AVCconstruction}
Let $\mathcal{O}$ be a finite outcome space with bounded payoff set
$\rho(\mathcal O)\subset\mathbb{R}^d$. If $C:\mathbb{R}^d\to\mathbb{R}$ is
closed and satisfies Axioms~\ref{axiom:price}--\ref{axiom:express}, then
there exists a strictly convex function
$R:\mathbb{R}^d \to (-\infty,+\infty]$ such that
\begin{equation}
\label{eq:dualcost}
C(q)
= R^*(q) =
\sup_{x\in \relint(H(\rho(\mathcal O)))}
q\cdot x - R(x)
\end{equation}
where $H(\rho(\mathcal O)) := \textrm{conv}(\{\rho (o):o\in \mathcal{O}\})$ is the convex hull over the payout-space.
Conversely, any strictly convex function $R$ defined on
$\relint(H(\rho(\mathcal O)))$ induces a cost function $C=R^*$ that
satisfies Axioms~\ref{axiom:price}--\ref{axiom:express}.
\end{theorem}

Besides the mentioned Axioms 1-4, it is also desirable that the market has a bounded worst-case loss. 

\begin{definition}[Worst-case loss]
\label{def:wcl}
Let $\Phi_T(o)$ denote the total payout owed by the mechanism at time $T$
under outcome $o \in \mathcal O$. For a sequence of trades $(r_t)_{t=1}^T$, the worst-case loss under outcome $o$ is
$\Loss_T(o)
:=
\Phi_T(o)
-
\sum_{t=1}^T \Pay_t(r_t)$.
\end{definition}

It so happens that a cost-function market maker constructed via Theorem \ref{thm:AVCconstruction} does have a bounded worst-case loss.
This is shown in Theorem \ref{prop:wcl-static}, which is simply true via Fenchel-Young inequality.

\begin{theorem}[Worst-case loss for a cost-function market]
\label{prop:wcl-static}
Let \(C:\mathbb{R}^d\to\mathbb{R}\) be a cost-function market constructed via Theorem \ref{thm:AVCconstruction}, with payment rule
$\Pay_t(r_t)=C(q_t)-C(q_{t-1})$.
Then for any horizon \(T\) and outcome \(o\in\mathcal O\),
\[
\Loss_T(o)
=
\rho(o)\cdot(q_T-q_0)-\sum_{t=1}^T \Pay_t(r_t)
\le
C^*(\rho(o)) + C(q_0)-\rho(o)\cdot q_0 .
\]
Consequently,
\[
\sup_{T,(r_t),o}\Loss_T(o)
\le
\max_{o\in\mathcal O} C^*(\rho(o))
+
C(q_0)-\min_{o\in\mathcal O}\rho(o)\cdot q_0 .
\]
In particular, if \(q_0=0\) and \(C(0)=0\), then $\sup_{T,(r_t),o}\Loss_T(o)
\le
\max_{o\in\mathcal O} C^*(\rho(o))$.
\end{theorem}

\begin{assumption}[cost-function construction]\label{ass:cost}
Throughout, we assume that the cost functions
$\{C_k : \mathbb{R}^d \to \mathbb{R}\}_{k=1}^M$ satisfy
Axioms~\ref{axiom:price}--\ref{axiom:express} with respect to a common
outcome space $\mathcal{O}$ and bounded payoff map
$\rho:\mathcal{O}\to\mathbb{R}^d$. Moreover, each $C_k$ admits a dual
representation $C_k = R_k^*$, where $R_k$ is strictly convex on
$\relint(H(\rho(\mathcal O)))$.
\end{assumption}

\subsection{Liquidity, Price Impact, and Market Objectives}

Liquidity governs how trades move prices. For a trade $r$ at state $q$,
the payment admits the decomposition 
\[
C(q+r)-C(q)
=
\nabla C(q)\cdot r
+
D_C(q+r,q),
\]
where $D_C$ is the Bregman divergence induced by $C$. The linear term
prices the trade at current marginal prices, while the divergence
captures price impact arising from curvature.

Thus, liquidity is encoded by curvature: higher curvature increases
price impact (lower liquidity), while flatter curvature reduces price
movement but increases exposure to inventory risk. A standard mechanism
for controlling liquidity is scaling via the perspective transform,
\[
C_\eta(q)=\eta\, C\!\left(\frac{q}{\eta}\right),
\qquad \eta>0,
\]
where larger $\eta$ corresponds to higher liquidity and smaller $\eta$
to lower liquidity. 

\subsection{Adaptive Liquidity as a Learning Problem}
\label{sec:adaptive-regret}

In practice, a market designer may wish to vary liquidity over time in
response to changing market conditions. For example, periods of heavy
trading may favor deeper liquidity to limit price impact, while quieter
periods may favor higher curvature to control inventory risk. 
We model this by introducing a family of candidate market makers
$\{C_k\}_{k=1}^M$, each representing a distinct liquidity regime, and
combining them using time-varying weights. This yields an online learning
formulation in which each regime acts as an expert and the weights are
updated based on observed outcomes.
The next section introduces our mechanism for combining these
cost-function markets.

\section{Adaptive Liquidity via Cost-Function Mixtures}
\label{sec:ensemble-diagonal}
Given cost functions $\{C_k : \mathbb{R}^d \to \mathbb{R}\}_{k=1}^M$ and weights $w \in \Delta_M $, define the mixed potential
\[
C^{\mix}(q;w)
=
\frac{1}{\beta}
\log\!\left(
\sum_{k=1}^M
w(k)\,e^{\beta C_k(q)}
\right),
\qquad \beta>0.
\]

This construction aggregates multiple liquidity regimes into a single convex potential.
Convexity in $q$ follows since the log-sum-exp map is convex and coordinatewise
nondecreasing, and each $C_k$ is convex.
The parameter $\beta$ controls how the mixture interpolates between regimes:
small $\beta$ yields near averaging, while large $\beta$ approaches a soft maximum.
The gradient admits the form
\[
\nabla C^{\mix}(q;w)
=
\sum_{k=1}^M
\pi_k(q;w)\,
\nabla C_k(q),
\qquad
\pi_k(q;w)
=
\frac{w(k)e^{\beta C_k(q)}}{\sum_j w(j)e^{\beta C_j(q)}}
\in \Delta_M,
\]
so prices are a soft aggregation of expert prices, while liquidity is governed
by the curvature $\nabla^2 C^{\mix}(q;w)$.
We define the payment rule
\[
\Pay_t^{\mathrm{mix}}(r_t)
=
C^{\mix}(q_t;w_{t+1})
-
C^{\mix}(q_{t-1};w_t)
+
\mathrm{fee}(q_{t-1},q_t,w_t,w_{t+1}),
\]
where the fee term is nonnegative and will be specified later.

The mechanism maintains a single global inventory $q_t \in \mathbb{R}^d$
while adapting liquidity through time-varying weights.
At each round $t$, trading proceeds under the potential
$C^{\mix}(\cdot; w_t)$, and the system evolves as follows:
\begin{enumerate}
\item The controller selects weights $w_t \in \Delta_M$.
\item A trader submits a trade $r_t$.
\item The inventory updates to $q_t = q_{t-1} + r_t$.
\item The weights are updated to $w_{t+1}$ based on observed feedback.
\item The payment $\Pay_t^{\mathrm{mix}}(r_t)$ is charged.
\end{enumerate}

Appendix~\ref{app:examples} provides examples showing how this construction
can be applied to standard LMSR markets and permutation-based pair-betting markets.



\section{Axioms Verification and Properties}

We verify that the adaptive mixture mechanism preserves the core economic
guarantees of cost-function markets under time-varying weights.
The key challenge is that weight updates change the underlying potential
across rounds, breaking standard telescoping arguments.
We address this via an explicit fee that compensates for switching costs,
and then establish no-arbitrage, bounded worst-case loss, expressiveness, and positive upside.
All proofs are deferred to Appendix~\ref{app:sec4}.

For this current section, we make the following assumption with respect to weights. 
\begin{assumption}[full support]
 Let $(w_t)_{t\ge1}$ be a sequence of weights in $\Delta_M$. 
We assume that weights have full support:
\begin{equation}\label{eq:full-support-aug}
w_t(k)>0
\qquad\forall\, t\ge 1,\ \forall\, k\in[M].
\end{equation}   
\end{assumption}
This assumption holds, for example, under multiplicative-weights updates.


\subsection{Fees and Switching Costs}
\label{sec:fees}

The payment rule for the mixed market includes an additive fee:
\[
\Pay_t^{\mathrm{mix}}(r_t)
=
C^{\mix}(q_t;w_{t+1})
-
C^{\mix}(q_{t-1};w_t)
+
\mathrm{fee}_t,
\]
where $\mathrm{fee}_t := \mathrm{fee}(q_{t-1},q_t,w_t,w_{t+1}) \ge 0$
compensates for changes in the potential due to weight updates.
We treat the fee abstractly, with the goal of preserving economic
properties such as no-arbitrage.
When weights change from $w_t$ to $w_{t+1}$, the potential
$C^{\mix}(\cdot;w)$ shifts. To compensate for worst-case decreases,
define the switch budget
\begin{equation}
\Delta_t^\star
:=
\sup_{q\in\mathbb{R}^d}
\big(
C^{\mix}(q;w_t)
-
C^{\mix}(q;w_{t+1})
\big).
\label{eq:switchpay}
\end{equation}
Since this quantity may be negative if the new potential dominates everywhere,
a natural choice of fee is its positive part,
\[
\mathrm{fee}_t = [\Delta_t^\star]_+ := \max\{\Delta_t^\star,0\}.
\]
This is the smallest nonnegative $\delta_t$ such that
\[
C^{\mix}(q;w_t)
\le
C^{\mix}(q;w_{t+1}) + \delta_t
\qquad \forall q.
\]
Thus, $[\Delta_t^\star]_+$ is a uniform \emph{insurance premium} against
worst-case decreases in the potential.
For a more in depth discussion regarding alternative fees and their tradeoffs, we refer the reader to Appendix \ref{app:fees}.

\subsection{No-Arbitrage under Dynamic Weights}
\label{sec:aug-master-noarb}

We first show that $\Pay_t^{\mathrm{mix}}$ satisfies no arbitrage. 
\begin{restatable}{theoremc}{noarb}
\label{thm:seq-noarb-dyn-mix}
Under Assumption~\ref{ass:cost} on $\{C_k\}_{k=1}^{M}$,
$\Pay_t^{\mathrm{mix}}$ with $\mathrm{fee}_t = [\Delta_t^\star]_+$ satisfies no-arbitrage.
\end{restatable}
Without compensation, changes in weights can decrease the potential between rounds, allowing traders to exploit the shift for risk-free profit. The fee term prevents this by ensuring each payment dominates a valid cost-function payment under the updated weights, so that, from the trader’s perspective, the market behaves like a sequence of valid cost-function markets stitched together.


\subsection{Bounded Worst-Case Loss}
\label{sec:wcl-switch-budget}

We next establish that adaptive mixtures inherit bounded worst-case loss from their constituent experts.
Adaptive updates change the pricing surface over time, so losses no
longer telescope directly as in static markets. In principle, repeated
weight changes could accumulate unbounded liability if not controlled.

\begin{restatable}{theoremc}{bwcl}
\label{thm:seq-bwcl-dyn-mix}
Under Assumption~\ref{ass:cost} on \(\{C_k\}_{k=1}^M\), for each $k\in [M]$ let $B_k := \max_{o\in\mathcal O} C_k^*(\rho(o))$. 
Then $\Pay_t^{\mathrm{mix}}$ with \(\mathrm{fee}_t=[\Delta_t^\star]_+\)
has bounded worst-case loss. In particular,
\[
\sup_{T,(r_t),o}\Loss_T(o)
\le
\max_{k\in[M]} B_k
+
\frac{1}{\beta}\log M
+
C^{\mix}(q_0;w_1)
-
\min_{o\in\mathcal O}\rho(o)\cdot q_0.
\]
\end{restatable}

The bound shows that the adaptive market inherits the worst-case loss
of the most conservative expert, up to an additive penalty of
$\frac{1}{\beta}\log M$. This term arises from the log-sum-exp
aggregation and reflects the cost of maintaining flexibility across
multiple liquidity regimes.
In particular, when $M$ is moderate and $\beta$ is not too small, this
overhead is negligible compared to the scale of the underlying market.


\subsection{Expressiveness of the Mixed Potential}

A prediction market must be able to represent arbitrary beliefs over
outcomes. If the adaptive mechanism restricted the range of achievable
prices, it would limit its usefulness as an information aggregation tool.

\begin{restatable}{theoremc}{expressivness}
\label{thm:expressive-master}
Under Assumption~\ref{ass:cost}, suppose the experts $\{C_k\}_{k=1}^M$
are generated from a common base cost $C$ via the scaled-family
construction $C_{\eta_k}(q) := \eta_k\,C(q/\eta_k)$.
Then the mixed potential $C^{\mix}(q;w)
:=
\frac{1}{\beta}\log\!\Big(\sum_{k=1}^M w(k)e^{\beta C_{\eta_k}(q)}\Big)$
is expressive on \(H(\rho(\mathcal O))\).
\end{restatable}

Despite combining multiple liquidity regimes, the mixture preserves the
full expressive power of the underlying cost-function family. In
particular, any interior belief can still be implemented as a market
state. Thus, adaptivity affects only the curvature (liquidity), not the
set of attainable prices.



\subsection{Positive Upside and Participation}
\label{sec:positive-upside}

In addition to no-arbitrage and bounded loss, a well-functioning market
should preserve participation incentives: traders who take directional
positions should be able to obtain strictly positive \emph{net profit}
when the corresponding outcome occurs.

Let $p_t
:=
\nabla C^{\mix}(q_{t-1};w_t)
\in
\relint H(\rho(\mathcal O))$
denote the current price vector. For a feasible direction \(v\), the
realized profit in outcome \(o\) from a trade \(r=s v\) is
$\rho(o)\cdot(sv)-\Pay_t^{\mix}(sv)$.

\begin{definition}[Positive directional upside]
\label{def:positive-upside}
The market satisfies \emph{positive directional upside} at round \(t\)
if, for every outcome \(o\in\mathcal O\) and every feasible direction
\(v\) satisfying
$(\rho(o)-p_t)\cdot v>0$, there exists a nontrivial interval $(s_{\min},s_{\max})$,
$0\le s_{\min}<s_{\max}$,
such that
\[
\rho(o)\cdot(sv)-\Pay_t^{\mix}(sv)>0
\qquad
\forall\, s\in(s_{\min},s_{\max}).
\]
\end{definition}

This definition allows profitable trades to occur away from the origin.
In particular, fixed or state-dependent fees may eliminate
profitability for very small trades, while curvature effects eventually
dominate sufficiently large trades. The economically relevant
requirement is therefore the existence of a \emph{nontrivial profitable
interval} rather than positivity only in the infinitesimal limit.

The following result shows that adaptive updates preserve such a
profitable region when the update-induced distortions remain controlled
relative to the directional payoff advantage.

\begin{restatable}[Positive directional upside under stable updates]{theorem}{positiveupside}
\label{thm:positive-upside}
Suppose the mixed cost \(C^{\mix}(\cdot;w_t)\) is locally smooth in
\(q\). Furthermore, assume that the weight-update dynamics and
adaptation fees remain sufficiently controlled relative to the
directional payoff advantage. Then \(\Pay_t^{\mix}\) admits positive
directional upside.
\end{restatable}

Theorem~\ref{thm:positive-upside} shows that realized profit is governed
by a competition between the linear payoff advantage
$(\rho(o)-p_t)\cdot v$
and the curvature and adaptation effects induced by the adaptive
mechanism. Fixed update and fee terms may introduce a minimum profitable
trade size, while curvature and dynamic distortions limit the maximum
profitable trade size. When the directional payoff advantage dominates
these effects, traders retain a nontrivial range of profitable trades.

Appendix~\ref{app:positive-upside-general} specializes this result to
LMSR mixtures. Corollary~\ref{cor:lmsr-positive-upside-general} shows
that, under stable adaptive updates, any directional payoff advantage
sufficiently separated from the current market price yields a nontrivial
interval of profitable trade sizes. In particular, for Arrow--Debreu
securities, any outcome \(o\) with \(p_t(o)<1\) admits a range of
purchase sizes for which buying additional \(o\)-shares yields strictly
positive realized profit when \(o\) occurs. The lower endpoint reflects
fixed update and fee effects, while the upper endpoint is determined by
curvature and adaptation costs. Thus, adaptive liquidity updates
preserve trader incentives over a nontrivial range of trade sizes rather
than only infinitesimally.

\section{Learning Hybrid Market Objectives}

We now analyze how learning over liquidity regimes translates into guarantees for the realized market. The challenge is that the learner optimizes a surrogate objective over experts, while the implemented market uses a nonlinear mixture with time-varying weights. Our strategy is therefore to relate these quantities through a sequence of comparisons, isolating the distortions introduced by mixing and adaptation. Appendix \ref{app:slippage-liability} shows that these distortions are controlled.

Let $\{C_k\}_{k=1}^M$ denote the candidate liquidity regimes. For a trade $r_t$ with $q_t = q_{t-1} + r_t$, define the slippage
\[
S_{k,t}
=
D_{C_k}(q_t, q_{t-1})
=
C_k(q_t) - C_k(q_{t-1}) - \nabla C_k(q_{t-1}) \cdot r_t,
\]
and the liability $L_k(q)
=
\max_{o\in\mathcal O}
\rho(o)\cdot q - \bigl(C_k(q) - C_k(0)\bigr)$.
The hybrid signal of expert $k$ is
\[
\Gamma^{\hyb}_{k,t}
=
a\Bigl(
S_{k,t}
-
\frac{1}{M}\sum_{j=1}^M S_{j,t}
\Bigr)
+
b L_k(q_t),
\qquad a,b>0.
\]
This signal encodes the liquidity trade-off: higher-liquidity regimes reduce slippage but increase inventory exposure, while lower-liquidity regimes do the opposite.

Given weights $w_t \in \Delta^M$, define the surrogate signal
$\Gamma^{\surr}_t
:=
\sum_{k=1}^M w_t(k)\Gamma^{\hyb}_{k,t}$.
We compare against the best sequence of regimes with at most $J$ switches, where for a sequence $j_{1:T}$,
$S(j_{1:T})
:=
\bigl|\{t\in\{2,\dots,T\}: j_t \neq j_{t-1}\}\bigr|$.
We assume the learner achieves sublinear switching regret
\[
R_T^{\surr}
:=
\sum_{t=1}^T \Gamma^{\surr}_t
-
\min_{j_{1:T}:S(j_{1:T})\le J}
\sum_{t=1}^T \Gamma^{\hyb}_{j_t,t}.
\]
This isolates the learning problem: we only require that the surrogate objective is optimized. 
We next relate this surrogate quantity to the mixed signal induced by the market, since this is the quantity that governs how trades are actually priced under the current mixture. 
By then incorporating weight updates and fees to obtain the realized signal, we show that good performance under the surrogate objective translates—up to controlled distortions—into good performance of the implemented mechanism.

The market evaluates trades under the mixed potential
$C_{\mix,t}(\cdot) := C_{\mix}(\cdot; w_t)$, inducing
\[
\Gamma^{\hyb}_{\mix,t}
=
a\Bigl(
D_{C_{\mix,t}}(q_t, q_{t-1})
-
\frac{1}{M}\sum_{j=1}^M S_{j,t}
\Bigr)
+
b L_{\mix,t}(q_t).
\]
To relate learning to the market, we compare $\Gamma^{\hyb}_{\mix,t}$ to $\Gamma^{\surr}_t$.
Since pricing uses posterior weights $\pi_t(\cdot;q)$ instead of $w_t$, the two differ. Theorem \ref{thm:hybrid-bridge-slippage} gives
\[
\Gamma^{\hyb}_{\mix,t}
\le
\Gamma^{\surr}_t
+
a (\pi_t - w_t)\cdot S_t
+
a \epsilon^{\text{drift}}_t.
\]

The first error term captures the \emph{pricing–learning mismatch}: the learner updates weights $w_t$, but prices are determined by the posterior weights $\pi_t$, which depend on the current state $q_t$. The second term $\epsilon^{\text{drift}}_t$ captures the fact that these posterior weights themselves change along the trade from $q_{t-1}$ to $q_t$. In other words, slippage is evaluated at $q_t$, while prices are determined at $q_{t-1}$, and the mixture weights shift between these two points. This produces a second-order correction term, which is small when trades are bounded and the cost functions are smooth.

The realized signal accounts for weight updates and fees:
\[
\Gamma^{\Phi,\hyb}_t
=
\Gamma^{\hyb}_{\mix,t}
+
a\bigl(C_{\mix}(q_t; w_{t+1}) - C_{\mix}(q_t; w_t)\bigr)
+
a\,\mathrm{fee}_t.
\]
Comparing this to $\Gamma^{\hyb}_{\mix,t}$ isolates implementation effects due to updating the mixture weights and compensating for switching costs. Theorem \ref{thm:transfer-clean} yields
\begin{align*}
\sum_{t=1}^T \Gamma^{\Phi,\hyb}_t
-
\sum_{t=1}^T \Gamma^{\hyb}_{j_t,t}
\le
R_T^{\surr}
&+
a \sum_{t=1}^T |(\pi_t - w_t)\cdot S_t|
\\
&+
a \sum_{t=1}^T \epsilon^{\text{drift}}_t
+
a \sum_{t=1}^T
\bigl(
\nabla_w C_{\mix}(q_t; w_t)\cdot (w_{t+1} - w_t)
\bigr)
+
a \sum_{t=1}^T \mathrm{fee}_t.
\end{align*}

Theorem \ref{thm:actual-hybrid-regret} shows that if $R_T^{\surr}$ is small and the update and fee terms are controlled, then
\[
\sum_{t=1}^T \Gamma^{\Phi,\hyb}_t
-
\sum_{t=1}^T \Gamma^{\hyb}_{j_t,t}
\le
R_T^{\surr}
+
a \sum_{t=1}^T |(\pi_t - w_t)\cdot S_t|
+
a \sum_{t=1}^T \epsilon^{\text{drift}}_t
+ C_w + C_f.
\]

Under dominance (Corollary \ref{cor:dominant-clean}), the mismatch terms vanish outside a burn-in set $B$, yielding
\[
\sum_{t=1}^T \Gamma^{\Phi,\hyb}_t
-
\sum_{t=1}^T \Gamma^{\hyb}_{j_t,t}
\le
R_T^{\surr}
+
4aGR C_{\text{dom}}
+
O(|B|).
\]
Thus, when regimes are well-separated, the adaptive market tracks the best sequence in hindsight and average regret vanishes.

\section{Simulation Study}
\label{sec:simulations}

\begin{figure}[!]
    \centering
    \includegraphics[scale=.47]{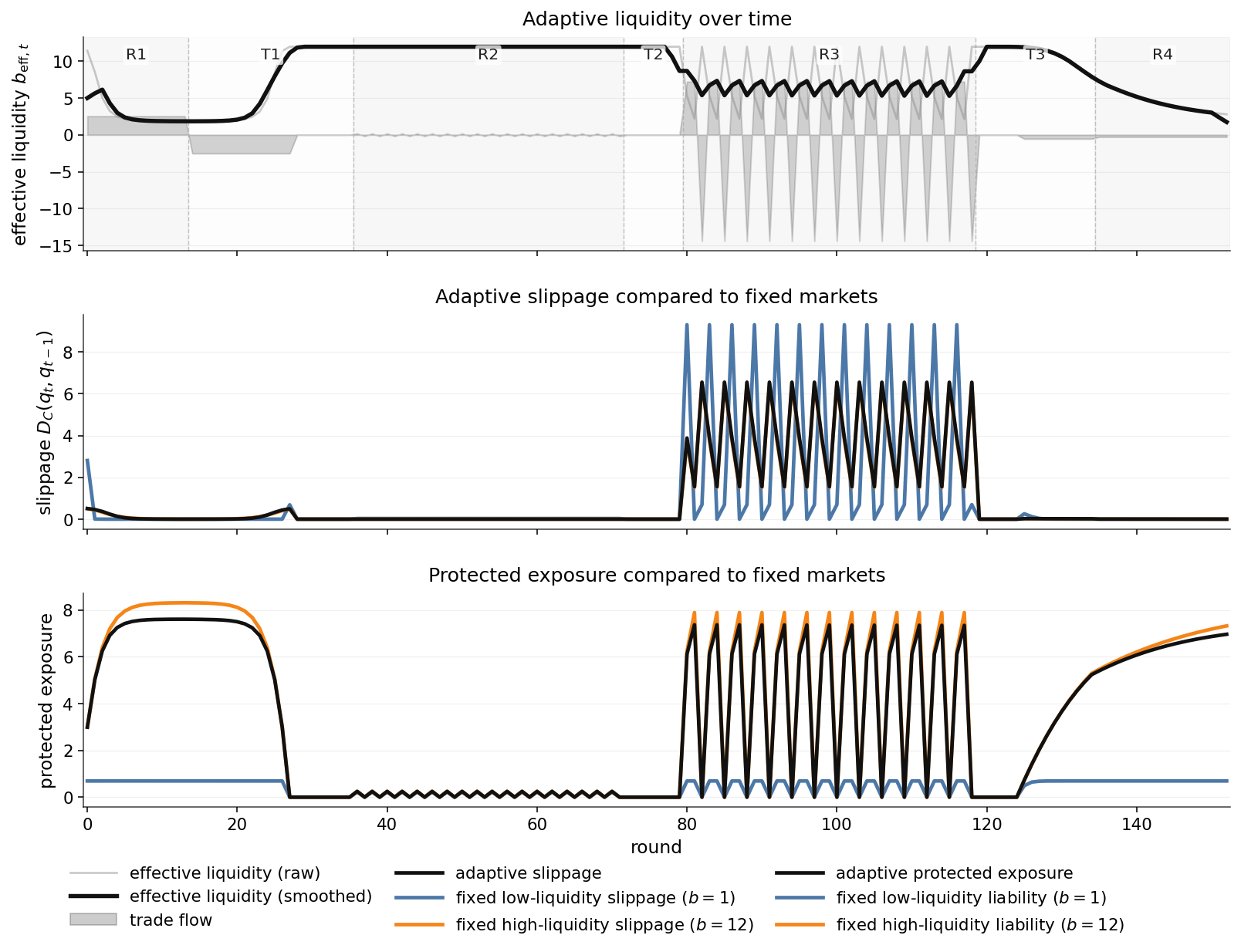}
\caption{Adaptive liquidity under regime changes. Top: effective liquidity
$b_{\mathrm{eff},t}
:=
\big(\sum_k \pi_t^{(b)}(k;q_t)\,\eta_k^{-1}\big)^{-1}$,
the curvature-matched scale of the mixture under pricing weights. Middle: mixture slippage
$D_{C_{\mix,t}}(q_t,q_{t-1})$ compared to fixed low- and high-liquidity markets. Bottom: mixture liability $L_{\mix,t}(q_t)$ compared to fixed markets, measuring worst-case
inventory exposure.
}
\label{fig:liquidity}
\end{figure}

We illustrate the adaptive mechanism in a two-outcome market with two
LMSR experts representing low and high liquidity. The learner updates
weights via Fixed-Share using the hybrid signal from Section~5, where
slippage \(S_{k,t}=D_{C_k}(q_t,q_{t-1})\) captures trading cost and
\(L_k(q_t)\) measures inventory risk --- detailed setup of the experiments can be found in Appendix ~\ref{app:code}\footnote{Code can be found at: \url{https://github.com/EnriqueNueve/Adaptive-Cost-Function-Liquidity-/tree/main}}.

The experiment is structured into four regimes designed to separate
price-impact and inventory effects.
Regimes R1 and R4 generate sustained
directional flow, leading to monotone inventory accumulation and hence
risk-dominated behavior. Regime R2 consists of small oscillatory trades
near neutral inventory, isolating the slippage component. Regime R3
captures price discovery with large alternating trades: flow is high,
but inventory remains bounded due to frequent reversals.

Figure~\ref{fig:liquidity} highlights three key behaviors.

First, effective liquidity adapts sharply to regime changes rather than
gradually. In R1, liquidity drops quickly to a low level after an initial
adjustment, reflecting rapid accumulation of inventory risk. During the
transition T1, liquidity increases abruptly and stabilizes at a high
level throughout R2, where inventory remains negligible. This plateau
indicates that when liability is inactive, the learner selects a
high-liquidity regime to minimize slippage.

Second, the comparison with fixed markets clarifies how adaptation
interpolates between extremes. In the slippage panel, the adaptive curve
tracks the low-liquidity market when trades are small (R2), but moves
closer to the high-liquidity market during large trades (R3). In R3,
slippage exhibits strong periodic spikes driven by alternating order
flow, and the adaptive mechanism sits consistently between the two fixed
benchmarks. This demonstrates that the mixture dynamically balances price
impact rather than committing to a single liquidity level.

Third, liability drives the main regime shifts. In the bottom panel,
protected exposure is large in R1 and R4 due to sustained inventory, and
nearly zero in R2 where trades cancel out. In R3, liability oscillates in
sync with trade flow but remains bounded, reflecting controlled exposure.
The adaptive mechanism closely tracks the higher-liability expert in
R1/R4 and suppresses exposure in R2, again interpolating between fixed
markets.

Overall, the behavior is consistent with the structure of the hybrid
signal. Slippage governs short-term responses to trade size, producing
the oscillatory adjustments seen in R3, while liability governs
long-term positioning, driving the sharp transitions between low- and
high-liquidity regimes. The resulting dynamics show that the learner does
not merely smooth between experts, but instead switches regimes in a
state-dependent manner, aligning liquidity with the dominant source of
risk at each point in time.


\section{Conclusion}
We introduced an adaptive prediction market that dynamically adjusts liquidity by combining multiple cost-function market makers through learnable weights. The resulting mechanism preserves the core guarantees of convex cost-function markets, including no-arbitrage, bounded worst-case loss, expressiveness, and positive upside while enabling liquidity to respond to changing market conditions.

Our analysis shows that adaptive liquidity can be framed as minimizing a structural risk signal that trades off price impact and inventory risk. Using this signal, the mechanism competes with the best sequence of liquidity regimes in hindsight, with deviations driven primarily by the mismatch between learning and pricing weights. When regimes are well-separated, this mismatch vanishes and the implemented market closely tracks the learned objective.

Our guarantees are expressed in terms of the hybrid structural signal
rather than direct economic welfare, and our experiments are limited to
stylized simulations. Important future directions include relating the
signal to objectives such as trader surplus and information aggregation,
studying equilibrium behavior under adaptive liquidity, and designing
alternative update mechanisms that reduce the mismatch between learning
and pricing weights. More broadly, this work shows that liquidity
selection can be treated as a learning problem, opening the door to
adaptive market designs that are both economically grounded and
statistically principled.


\newpage 
\bibliographystyle{plainnat}
\bibliography{bib}


\newpage 


\appendix
\section{Proofs and Additional Results from Section 3}

\begin{lemma}[Gradient and Hessian of the log-sum-exp mixed potential]
\label{lem:master-grad-hess-varcov}
Fix $\beta>0$ and $w\in\Delta_M$ with full support.  
Let $C_1,\dots,C_M:\R^d\to\R$ be twice continuously differentiable and define $C^{\mix}:= \frac{1}{\beta}\log \sum_{k=1}^M w (k) e^{\beta C_k(q)}$. Then the gradient and Hessian of the mixed potential are
\begin{align}
\nabla C^{\mix}(q;w)
&=
\sum_{k=1}^M \pi_k(q;w)\,\nabla C_k(q),
\label{eq:master-grad}\\[6pt]
\nabla^2 C^{\mix}(q;w)
&=
\sum_{k=1}^M \pi_k(q;w)\,\nabla^2 C_k(q)
+
\beta\,\Cov_{k\sim\pi(q;w)}\!\big(\nabla C_k(q) ; q, w\big),
\label{eq:master-hess}
\end{align}
where the covariance of the gradients under the posterior distribution is defined as:
\[
\Cov_{k\sim\pi(q;w)}[\nabla C_k(q)]
:=
\sum_{k=1}^M
\pi_k\,(\nabla C_k(q)- \nabla_q C^{\mix}(q;w))(\nabla C_k(q)-\nabla_q C^{\mix}(q;w))^\top
\succeq 0 .
\]

Then, for any direction $r \in \mathbb{R}^d$, the directional curvature of the mixed potential decomposes into the expected curvature of individual experts and the variance of their directional derivatives:
\begin{equation}
\label{eq:master-hess-rayleigh}
r^\top \nabla^2 C^{\mix}(q;w) r
=
\sum_{k=1}^M
\pi_k(q;w)\, r^\top \nabla^2 C_k(q) r
+
\beta\,\Var_{k\sim\pi(q;w)}\!\big(r^\top \nabla C_k(q)\big),
\end{equation}
where $\Var_{k\sim\nu}(X):=\E_{k\sim\nu}[(X-\E X)^2]$.
\end{lemma}
\begin{proof}
We define the partition function $Z(q)$ such that the mixed potential $C^{\text{mix}}(q; w)$ is exactly the scaled log-partition function:
\[
Z(q)
:=
\sum_{k=1}^M w (k) e^{\beta C_k(q)},
\qquad
C^{\mix}(q;w)=\frac{1}{\beta}\log Z(q).
\]

Since the sum defining $Z(q)$ is finite and each $C_k$ is twice continuously differentiable, we may
differentiate term-by-term.

We first derive the gradient. By the chain rule, $\nabla C^{\text{mix}}(q; w) = \frac{1}{\beta} \frac{\nabla Z(q)}{Z(q)}$. Computing the gradient $\nabla Z(q)$ and factoring out $Z(q)$ yields
$\nabla Z(q)
=\sum_{k=1}^M w(k) e^{\beta C_k(q)}\,\beta\nabla C_k(q)
=\beta Z(q)
\sum_{k=1}^M \pi_k(q;w)\nabla C_k(q)$.
Substituting into the chain rule expression gives
$\nabla C^{\mix}(q;w)
= \sum_{k=1}^M \pi_k(q;w)\nabla C_k(q)$, which yields exactly \eqref{eq:master-grad}.

We now solve for the derivative of the posterior weights.
We express the unnormalized weights as $a_k(q):=w (k) e^{\beta C_k(q)}$ such that $\pi_k=\frac{a_k}{Z}$. Applying the chain rule gives
\[\nabla a_k(q)
=a_k(q)\,\beta\nabla C_k(q)\]
By linearity, we have
\[\nabla Z(q)=
\sum_{j=1}^M\nabla a_j(q)=
\beta Z(q)\nabla C^{\mix}(q;w).\]
Using the quotient rule to $\pi_k$ yields 
\[\nabla \pi_k
=\frac{\nabla a_k}{Z}-\frac{a_k}{Z^2}\nabla Z=
\beta\pi_k\big(\nabla C_k(q)- \nabla C^{\mix}(q;w) \big).\]

We now solve for the Hessian.
Differentiating \eqref{eq:master-grad} and applying the product rule gives
\[
\nabla^2 C^{\mix}(q;w)
=
\sum_{k=1}^M
\Big(
(\nabla \pi_k)\nabla C_k(q)^\top
+
\pi_k\nabla^2 C_k(q)
\Big).\]
Substituting the expression for $\nabla\pi_k$ yields
\[\nabla^2 C^{\mix}(q;w)
=
\sum_{k=1}^M
\pi_k\nabla^2 C_k(q)
+
\beta
\sum_{k=1}^M
\pi_k(\nabla C_k(q)- \nabla C^{\mix}(q;w) )\nabla C_k(q)^\top.\]
To rewrite the second term as a covariance matrix, let $\bar{g}(q) := \nabla C^{\text{mix}}(q; w)$ be the expected gradient. By definition, the expected centered gradient is zero $$\sum_{k=1}^M \pi_k (\nabla C_k(q) - \bar{g}(q)) = 0$$
We introduce $\bar{g}(q)$ into the outer product
\[\sum_{k=1}^M
\pi_k(\nabla C_k-\bar g)\nabla C_k^\top
=\sum_{k=1}^M \pi_k(\nabla C_k - \bar{g})(\nabla C_k - \bar{g})^\top
= \sum_{k=1}^M
\pi_k(\nabla C_k- \nabla C^{\mix}(q;w) )(\nabla C_k- \nabla C^{\mix}(q;w))^\top.\]

Substituting this back into the Hessian expression becomes
\[
\nabla^2 C^{\mix}(q;w)
=
\sum_{k=1}^M
\pi_k(q;w)\nabla^2 C_k(q)
+
\beta
\sum_{k=1}^M
\pi_k(q;w)
(\nabla C_k(q)-\nabla C^{\mix}(q;w) )(\nabla C_k(q)- \nabla C^{\mix}(q;w) )^\top,
\]
which is exactly \eqref{eq:master-hess}.

Finally, we derive the Hessian's quadratic form.
For any vector $r\in\R^d$, pre- and post-multiplying the covariance matrix by $r$ isolates the variance along the direction of the trade
\[
r^\top
\Cov_{k\sim\pi(q;w)}[\nabla C_k(q)]
r
=
\E_{k\sim\pi}\big[(r^\top(\nabla C_k(q)- \nabla C^{\mix}(q;w)))^2\big]
=
\Var_{k\sim\pi(q;w)}[r^\top \nabla C_k(q)].\]

Applying this identity to \eqref{eq:master-hess} yields
\eqref{eq:master-hess-rayleigh}.
\end{proof}

\begin{theorem}[Smoothness of the log-sum-exp mixed potential]
\label{thm:smoothness-logsum-master}
Fix $\beta>0$ and $w\in\Delta_M$ with full support. Let $C_1,\dots,C_M:\mathbb R^d\to\mathbb R$ be convex and twice continuously differentiable and $\beta>0$, and define $C^{\mix}:= \frac{1}{\beta}\log \sum_{k=1}^M w (k) e^{\beta C_k(q)}$. Let $\|\cdot\|$ be a norm in $\mathbb{R}^d$ with dual norm $\|\cdot\|_*$. Assume
\begin{itemize}
    \item[(i)] Each $C_k$ is $L_k$-smooth with respect to $\|\cdot\|$
    \item[(ii)] For all $k \in [M],q \in \mathbb{R}^d$, $\|\nabla C_k(q)\|_*\le G$
\end{itemize}
Then $C^{\mix}(\cdot;w)$ is $L_{\mix}$-smooth w.r.t.\ $\|\cdot\|$ with
\[L_{\mix}\le L_{\max}+\beta G^2,
\qquad L_{\max}:=\max_{k} L_k.\]
\end{theorem}
\begin{proof}
Via Lemma \ref{lem:master-grad-hess-varcov}, for any direction $r\in\mathbb R^d$, the quadratic form of the mixed Hessian is given as
\begin{align*}
r^\top \nabla^2 C^{\mix}(q;w)\, r
&=
\sum_{k=1}^M \pi_k(q;w)\, r^\top \nabla^2 C_k(q)\, r
\;+\;
\beta\sum_{k=1}^M \pi_k(q;w)\Big(r\cdot(\nabla C_k(q)-\nabla C^{\mix}(q;w))\Big)^2 \\
&=
\sum_{k=1}^M \pi_k(q;w)\, r^\top \nabla^2 C_k(q)\, r
\;+\;
\beta\,\Var_{k\sim\pi(q;w)}\!\big(r\cdot \nabla C_k(q)\big),
\end{align*}
where $
\Var_{k\sim\pi(q;w)}\!\big(r\cdot \nabla C_k(q)\big)
:=
\sum_{k=1}^M \pi_k(q;w)\Big(r\cdot \nabla C_k(q)-r\cdot \nabla C^{\mix}(q;w)\Big)^2$.
\newline We bound the two terms separately. 
The assumption that each $C_k$ is $L_k$-smooth gives \[r^\top \nabla^2 C_k(q)\,r\le L_k\|r\|^2.\] 
Since the posterior probabilities $\pi_k$ form a valid distribution, the convex combination is bounded by the maximum component: $\sum_k \pi_k r^\top \nabla^2 C_k(q)\,r \le L_{\max}\|r\|^2$.
For the variance term, by definition of the dual norm, we have
\[\big|r\cdot \nabla C_k(q)\big| \le
\|r\|\,\|\nabla C_k(q)\|_*
\le \|r\|\,G.\]
This implies $(r\cdot \nabla C_k(q))^2\le \|r\|^2G^2$. For any random variable $X$, $\text{Var}(X) \le \mathbb{E}[X^2]$, we bound the variance strictly by its second moment
\[\Var_{k\sim\pi}\!\big(r\cdot \nabla C_k(q)\big) \le \sum_{k=1}^M \pi_k (r\cdot \nabla C_k(q))^2 \le \|r\|^2 G^2.\]
Combining these bounds yields
\[
r^\top \nabla^2 C^{\mix}(q;w)\, r
\le
\big(L_{\max}+\beta G^2\big)\|r\|^2
\quad \forall r \in \mathbb{R}^d,\]
which exactly satisfies the definition of $L_{\text{mix}}$-smoothness with $L_{\mix}\le L_{\max}+\beta G^2$.
\end{proof}

\begin{lemma}[First-order upper bound in the weights]
\label{lem:cmix-concave-w}
Fix \(q\in\mathbb{R}^d\) and \(\beta>0\) and define $C^{\mix}(q;w)
=
\frac1\beta
\log\!\Big(\sum_{k=1}^M w (k) e^{\beta C_k(q)}\Big)$.
Let $\Delta_\delta^M := \{w \in \Delta_M : w_k \ge \delta, \forall  k \in [M]\}$ for some $\delta > 0$. For any $w_t, w_{t+1} \in \Delta_\delta^M$, the map $w \mapsto C^{\text{mix}}(q; w)$ is concave, and specifically satisfies the upper bound:
\[C^{\mix}(q;w_{t+1})-C^{\mix}(q;w_t)
\le
\nabla_w C^{\mix}(q;w_t)\cdot (w_{t+1}-w_t).\]
\end{lemma}

\begin{proof}
Let $a(q) \in \mathbb{R}^M$ denote the vector with strictly positive components $a_k(q) := e^{\beta C_k(q)} > 0$. We rewrite the partition function $Z(q, w) = w^\top a(q)$. Then the mixed potential becomes
\[C^{\mix}(q;w)=\frac1\beta \log Z(q,w).\]
Since for all $k \in [M]$, \(a_k(q)>0\), the partition function is strictly positive ($w^\top a(q) > 0$), so \(Z(q,w)>0\), for all \(w\in\Delta_\delta^M\). Thus, $C^{\text{mix}}(q; \cdot)$ is infinitely differentiable on the open set \(\{w:Z(q,w)>0\}\) containing $\Delta_M$.

Differentiating with respect to vector \(w\) yields the gradient
\[
\nabla_w C^{\mix}(q;w)=\frac1\beta \frac{a(q)}{Z(q,w)}\]
Differentiating again yields the Hessian matrix
\[
\nabla_w^2 C^{\mix}(q;w)
=
-\frac1\beta \frac{a(q)a(q)^\top}{Z(q,w)^2}\]
For any vector $u \in \mathbb{R}^M$, the quadratic form of the Hessian is exactly
\[u^\top \nabla_w^2 C_{\mix}(q;w)u
=
-\frac1\beta \frac{(u^\top a(q))^2}{Z(q,w)^2}
\le 0.\]
Because the Hessian is negative semi-definite, \(\nabla_w^2 C^{\mix}(q;w)\preceq 0\), the map $w \mapsto C_{\text{mix}}(q; w)$ is globally concave on its domain, implying it is concave on $\Delta_\delta^M$.
Applying the supporting-hyperplane inequality for concave differentiable functions
\[
C^{\mix}(q;v)
\le
C^{\mix}(q;u)+\nabla_w C^{\mix}(q;u)\cdot (v-u), \quad \forall u,v\in\Delta_\delta^M \]
Setting \(u=w_t\) and \(v=w_{t+1}\) yields the desired first-order bound:
\[
C^{\mix}(q;w_{t+1})-C^{\mix}(q;w_t)
\le
\nabla_w C^{\mix}(q;w_t)^\top (w_{t+1}-w_t).
\]
\end{proof}

\begin{corollary}[Cost-update bound]\label{corr:costupdatedbound}
For each round \(t\) and strictly positive scalar $a\in\reals_{>0}$,
\[a\big(C^{\mix}(q_t;w_{t+1})-C^{\mix}(q_t;w_t)\big)
\le
a\,\nabla_w C^{\mix}(q_t;w_t)\cdot (w_{t+1}-w_t).\]
Consequently, bounding this by its positive part yields
\[a\big(C^{\mix}(q_t;w_{t+1})-C^{\mix}(q_t;w_t)\big)
\le
a\Big[\nabla_w C^{\mix}(q_t;w_t)\cdot (w_{t+1}-w_t)\Big]_+ .\]
\end{corollary}
\begin{proof}
The first inequality follows directly from multiplying the concavity upper bound shown in Lemma \ref{lem:cmix-concave-w} by $a\in\reals_{>0}$. The second inequality follows trivially as any real number is bounded by its positive part, $x \le [x]_+$. 
\end{proof}


\section{Switching Fees and Budget Corrections}
\label{app:fees}
When mixture weights change from $w_t$ to $w_{t+1}$, the potential
$C^{\mix}(\cdot;w)$ changes. As a result, the standard telescoping
structure of cost-function markets no longer holds automatically, and
payments must be corrected to preserve no-arbitrage and bounded loss.

The adaptive payment rule is
\begin{equation}
\Pay_t^{\mix}(r_t)
=
C^{\mix}(q_t;w_{t+1})
-
C^{\mix}(q_{t-1};w_t)
+
\mathrm{fee}_t,
\label{eq:payment-fee}
\end{equation}
where $\mathrm{fee}_t$ compensates for the change in the potential induced by
the weight update.

The minimal globally valid positive correction is given by the 
switch budget
\begin{equation}
[\Delta_t^\star ]_+
:=
\sup_{q\in\mathbb{R}^d}
\Big(
C^{\mix}(q;w_t)
-
C^{\mix}(q;w_{t+1})
\Big)_+.
\label{eq:delta-star}
\end{equation}
This is the smallest quantity
such that
\[
C^{\mix}(q;w_t)
\le
C^{\mix}(q;w_{t+1})
+
[\Delta_t^\star]_+
\quad
\text{for all } q\in\mathbb{R}^d.
\]

By construction, $[\Delta_t^\star]_+ = 0$ whenever the
weight update does not decrease the potential anywhere. Thus,
$[\Delta_t^\star]_+$ acts as a one-sided correction, compensating only for
adverse changes in the potential. Substituting
$\textrm{fee}_t = [\Delta_t^\star]_+$ into \eqref{eq:payment-fee} yields
\[
\sum_{t=1}^T \Pay_t^{\mix}(r_t)
=
C^{\mix}(q_T;w_{T+1})
-
C^{\mix}(q_0;w_1)
+
\sum_{t=1}^T [\Delta_t^\star]_+,
\]
so the cumulative effect of weight updates is captured by the switching
budget $\sum_{t=1}^T [\Delta_t^\star]_+$, which ensures bounded loss.

The quantity $\Delta_t^\star$ can be interpreted as an insurance premium
that protects against worst-case decreases in the potential when
switching between liquidity regimes. When weight updates are small, the
difference $C^{\mix}(q;w_t)-C^{\mix}(q;w_{t+1})$ is small uniformly in
$q$, and thus $[\Delta_t^\star]_+$ is negligible.

The definition \eqref{eq:delta-star} takes a supremum over all
$q\in\mathbb{R}^d$, which can be conservative when trades are bounded.
If trades satisfy $\|r_t\|\le R$, then the reachable states lie in a set
$\mathcal Q_R$, and one may instead consider the restricted switch
budget
\begin{equation}
[\Delta_t^{\mathrm{res}}]_+
:=
\sup_{q\in \mathcal Q_R}
\Big(
C^{\mix}(q;w_t)
-
C^{\mix}(q;w_{t+1})
\Big)_+,
\label{eq:delta-restricted}
\end{equation}
which is typically smaller than $[\Delta_t^\star]_+$. This yields the weaker
guarantee
\[
C^{\mix}(q;w_t)
\le
C^{\mix}(q;w_{t+1})
+
\Delta_t^{\mathrm{res}}
\quad
\text{for all reachable } q,
\]
but may fail outside $\mathcal Q_R$.

An alternative is to use a pathwise correction depending only on the
realized trade,
\begin{equation}
[\Delta_t^{\mathrm{path}}(r_t)]_+
:=
\Big(
C^{\mix}(q_t;w_t)
-
C^{\mix}(q_t;w_{t+1})
\Big)_+.
\label{eq:delta-path}
\end{equation}
This ensures that the correction is exact along the realized trajectory,
so that payments telescope pathwise. However, it does not provide a
uniform bound over all states and therefore does not guarantee
no-arbitrage under arbitrary counterfactual trades.

These choices reflect a tradeoff between robustness and efficiency. The
global correction $[\Delta_t^\star]_+$ provides the strongest guarantees but
can be conservative, while restricted and pathwise corrections reduce the
magnitude of fees at the cost of weaker guarantees.

\section{Proofs and Additional Results from Section 4}\label{app:sec4}

\begin{theorem}
\label{thm:switch-fee-domination}
Assume the weight sequence satisfies the full support condition \eqref{eq:full-support-aug}.
Then for every horizon $T \ge 1$, it holds for $\Pay_t^{\mathrm{mix}}$ with $\mathrm{fee}_t = [\Delta_t^\star]_+$
\begin{equation}
\sum_{t=1}^T \Pay_t^{\mathrm{mix}}(r_t)
\ge
C^{\mathrm{mix}}(q_T;w_1)
-
C^{\mathrm{mix}}(q_0;w_1).
\label{eq:pay-lowerbound-fee}
\end{equation}

\end{theorem}

\begin{proof} To prove our result, we do so via three steps.
\newline \textbf{1) One-step domination.}
By the definition of the supremum $\Delta_t^\star$, for all $q \in \mathbb{R}^d$, we have
\[
C^{\mathrm{mix}}(q;w_t)
\le
C^{\mathrm{mix}}(q;w_{t+1})
+
\Delta_t^\star.
\]
Evaluating this inequality at $q = q_{t-1}$ gives
\[
C^{\mathrm{mix}}(q_{t-1};w_t)
\le
C^{\mathrm{mix}}(q_{t-1};w_{t+1})
+
\Delta_t^\star.
\]
Because the fee is defined as the positive part of the cost increase, $\fee_t = [\Delta_t^\star]_+ \ge \Delta_t^\star$, we write
\[
C^{\mathrm{mix}}(q_{t-1};w_t)
\le
C^{\mathrm{mix}}(q_{t-1};w_{t+1})
+
\fee_t.\]
Rearranging the terms gives
\[
-
C^{\mathrm{mix}}(q_{t-1};w_t)
+
\fee_t
\ge
-
C^{\mathrm{mix}}(q_{t-1};w_{t+1}).
\]
Substituting this lower into the definition of the payment rule yields
\begin{align*}
\Pay_t^{\mathrm{mix}}(r_t)
&=
C^{\mathrm{mix}}(q_t;w_{t+1})
-
C^{\mathrm{mix}}(q_{t-1};w_t)
+
\fee_t \\
&\ge
C^{\mathrm{mix}}(q_t;w_{t+1})
-
C^{\mathrm{mix}}(q_{t-1};w_{t+1}).
\end{align*}

\medskip
\noindent
\textbf{2) Uniform domination.}
From the definition of $\Delta_t^\star$, for all $q \in \mathbb{R}^d$, we know
\[C^{\text{mix}}(q; w_t) - C^{\text{mix}}(q; w_{t+1}) \le \Delta_t^\star.\]
Rearranging the terms yields
\[
C^{\mathrm{mix}}(q;w_t)
\le
C^{\mathrm{mix}}(q;w_{t+1})
+
\Delta_t^\star.\]
Summing these inequalities over all rounds $t = 1, \dots, T$ causes the intermediate cost terms to telescope, yielding
\[
C^{\mathrm{mix}}(q;w_1)
\le
C^{\mathrm{mix}}(q;w_{T+1})
+
\sum_{t=1}^T \Delta_t^\star,
\]
which provides the uniform domination bound for any $q$.

\medskip
\noindent
\textbf{3) Exact telescoping and lower bound.}
Summing the payment rule over $T$ rounds gives
\begin{align*}
\sum_{t=1}^T \Pay_t^{\mathrm{mix}}(r_t)
&=
\sum_{t=1}^T
\Big(
C^{\mathrm{mix}}(q_t;w_{t+1})
-
C^{\mathrm{mix}}(q_{t-1};w_t)
+
\fee_t
\Big).
\end{align*}
Observe the cost terms form an exact telescoping sum
\[
\sum_{t=1}^T
\Big(
C^{\mathrm{mix}}(q_t;w_{t+1})
-
C^{\mathrm{mix}}(q_{t-1};w_t)
\Big)
=
C^{\mathrm{mix}}(q_T;w_{T+1})
-
C^{\mathrm{mix}}(q_0;w_1),
\]
Adding the cumulative fee $\sum_{t=1}^T \fee_t$ back to this result 
\[
\sum_{t=1}^T \Pay_t^{\mathrm{mix}}(r_t)
=
C^{\mathrm{mix}}(q_T;w_{T+1})
-
C^{\mathrm{mix}}(q_0;w_1)
+
\sum_{t=1}^T \fee_t.
\]
Lastly, to establish the cumulative lower bound, we apply the uniform domination inequality (Step 2) specifically at $q = q_T$, we have
\[
C^{\mathrm{mix}}(q_T;w_{T+1})
+
\sum_{t=1}^T \fee_t
\ge
C^{\mathrm{mix}}(q_T;w_1).
\]
Since $\fee_t \ge \Delta_t^\star$ for all $t$, substituting the fees yields
\[C^{\text{mix}}(q_T; w_{T+1}) + \sum_{t=1}^T \fee_t \ge C^{\text{mix}}(q_T; w_1).\] 
Substituting this lower bound into our exact telescoping equality yields
\[\sum_{t=1}^T \Pay_t^{\mathrm{mix}}(r_t)
\ge
C^{\mathrm{mix}}(q_T;w_1)
-
C^{\mathrm{mix}}(q_0;w_1),
\]
which proves \eqref{eq:pay-lowerbound-fee} and completes the proof.
\end{proof}

We provide the following proofs regarding the no-arbitrage guarantees and worst-case loss bounds of the mixed payment rule. We recall Assumption 1 and formalize its relationship to the finite outcome space $\mathcal{O}$. Assumption 1 guarantees that the gradients of the expert cost functions always reside within the convex hull $H(\rho(\mathcal{O}))$.
\noarb* 
\begin{proof}
By Assumption 1, the gradients of the expert cost functions lie within the convex hull of the outcome payoff vectors $\nabla C_k(q) \in H(\rho(\mathcal{O}))$ for all $k$. For any fixed $w\in\Delta_M$, the mixed potential $C^{\mathrm{mix}}(\cdot;w)$ is convex and
differentiable, and its gradient is a convex combination of expert gradients:
\begin{equation}
\nabla_q C^{\mathrm{mix}}(q;w)
=
\sum_{k=1}^M \pi_k(q;w)\,\nabla C_k(q),
\qquad
\pi_k(q;w)
=
\frac{w(k)e^{\beta C_k(q)}}{\sum_j w(j)e^{\beta C_j(q)}}
\in\Delta_M.
\label{eq:grad-mix}
\end{equation}
Because $\nabla_q C^{\text{mix}}(q; w)$ is a convex combination of vectors that reside in $H(\rho(\mathcal O))$, the mixed gradient itself remains in the convex hull $H(\rho(\mathcal O))$
\[\nabla_q C^{\mathrm{mix}}(q;w)\in H(\rho(\mathcal O))
\quad \forall q,\ \forall w.\]
Fix any trade sequence of length $T$ and denote the net change in shares as $\Sigma = q_T - q_0$.
By Theorem~\ref{thm:switch-fee-domination}, the cumulative payments are bounded below by the difference in the initial potential
\begin{equation}
\sum_{t=1}^T \Pay_{t}^{\mathrm{mix}}(r_t)
\ge
C^{\mathrm{mix}}(q_T;w_1)-C^{\mathrm{mix}}(q_0;w_1).
\label{eq:pay-dominate-fixedmaster}
\end{equation}
Since $C^{\mathrm{mix}}(\cdot;w_1)$ is differentiable, applying the fundamental theorem of calculus along the line segment
$q_0+\tau\Sigma$ for $\tau\in[0,1]$, gives
\begin{equation}
C^{\mathrm{mix}}(q_T;w_1)-C^{\mathrm{mix}}(q_0;w_1)
=
\int_0^1
\nabla_q C^{\mathrm{mix}}(q_0+\tau\Sigma;w_1)\cdot \Sigma\, d\tau.
\label{eq:ftc-line}
\end{equation}
For each $\tau$, the vector $\nabla_q C^{\text{mix}}(q_0 + \tau \Sigma; w_1)$ lies in $H(\rho(\mathcal O))$. Thus, we express the vector as a convex combination of the outcome payoffs, $\sum_{o \in \mathcal O} \alpha_o \rho(o)$, where $\alpha_o \ge 0$ and $\sum_{o \in \mathcal O} \alpha_o = 1$. Consequently, its inner product with the net trade $\Sigma$ is strictly bounded below by the minimum possible inner product over all individual outcomes:
\begin{equation}
\nabla_q C^{\mathrm{mix}}(q_0+\tau\Sigma;w_1)\cdot \Sigma = \sum_{o \in \mathcal O} \alpha_o \left(\rho(o) \cdot \Sigma\right)
\ge \min_{o\in\mathcal O}\rho(o)\cdot \Sigma
\quad \forall \tau\in[0,1].
\label{eq:hull-min}
\end{equation}
Integrating this uniform lower bound from \eqref{eq:hull-min} over $\tau\in[0,1]$ and using \eqref{eq:ftc-line} yields
\[
C^{\mathrm{mix}}(q_T;w_1)-C^{\mathrm{mix}}(q_0;w_1)
\ge
\min_{o\in\mathcal O}\rho(o)\cdot \Sigma.
\]
Combining this with the telescoping bound from \eqref{eq:pay-dominate-fixedmaster} gives
\[
\sum_{t=1}^T \Pay_{t}^{\mathrm{mix}}(r_t)
\ge
\min_{o\in\mathcal O}\rho(o)\cdot \Sigma.
\]
Choosing the worst-case realized outcome $o^\star \in \arg\min_{o\in\mathcal O} \rho(o) \cdot \Sigma$ establishes that the total payment strictly bounds the worst-case payout, proving no-arbitrage. 
\end{proof}


\bwcl* 
\begin{proof} 
Fix any horizon $T$, trade sequence $(r_t)$, and outcome $o \in \mathcal O$.

We first show that we can upper bound the total loss in terms of a bound on the payment sum.
By definition of market loss and the cumulative lower bound established in Theorem \ref{thm:switch-fee-domination}, we have
\begin{align}
\Loss_T(o)
&=
\rho(o)\cdot(q_T-q_0)
-
\sum_{t=1}^T \Pay_{t}^{\mix}(r_t)
\nonumber\\
&\le
\rho(o)\cdot(q_T-q_0)
-
\big(
C^{\mix}(q_T;w_1)
-
C^{\mix}(q_0;w_1)
\big).
\nonumber
\end{align}
Rearranging gives
\begin{equation}
\Loss_T(o)
\le
\underbrace{\big(\rho(o)\cdot q_T - C^{\mix}(q_T;w_1)\big)}_{(A)}
+
\underbrace{\big(C^{\mix}(q_0;w_1)-\rho(o)\cdot q_0\big)}_{(B)}.
\label{eq:loss-split-explicit}
\end{equation}

We first bound term (A).
Let $w_1\in\Delta_M$.
Since $\sum_{k=1}^M w_1(k)=1$, there exists at least one index $k^\star$ such that $w_1(k^\star)\ge 1/M$ via pigeonhole argument. For any state $q \in \mathbb{R}^d$, we can lower-bound the mixed potential
\begin{align}
C^{\mix}(q;w_1)
&=
\frac{1}{\beta}
\log\!\Big(\sum_{k=1}^M w_1(k)e^{\beta C_k(q)}\Big)
\nonumber\\
&\ge
\frac{1}{\beta}
\log\!\Big(w_1(k^\star)e^{\beta C_{k^\star}(q)}\Big)
\nonumber\\
&=
C_{k^\star}(q)
+
\frac{1}{\beta}\log w_1(k^\star)
\nonumber\\
&\ge
C_{k^\star}(q)
-
\frac{1}{\beta}\log M.
\label{eq:cmix-lower-bound}
\end{align}
Therefore, we have
\begin{align}
\rho(o)\cdot q
-
C^{\mix}(q;w_1)
&\le
\rho(o)\cdot q
-
C_{k^\star}(q)
+
\frac{1}{\beta}\log M.
\nonumber
\end{align}
By the definition of the convex conjugate (Fenchel-Young inequality), we know
$$\rho(o) \cdot q - C_{k^\star}(q) \le C_{k^\star}^*(\rho(o)).$$ 
Then for all $q \in \mathbb{R}^d$ and $o \in \mathcal{O}$, term (A) is uniformly bounded
\[
(A)
\le
C_{k^\star}^*(\rho(o))
+
\frac{1}{\beta}\log M
\le
\max_{k\in[M]} B_k
+
\frac{1}{\beta}\log M.
\]

We now bound term (B).
Term (B) depends strictly on the initial state and can be bounded by taking the worst-case outcome
\[
(B)
=
C^{\mix}(q_0;w_1)
-
\rho(o)\cdot q_0
\le
C^{\mix}(q_0;w_1)
-
\min_{o\in\mathcal O}\rho(o)\cdot q_0.
\]

We now combine our bounds.
Substituting the uniform bounds for (A) and (B) back into \eqref{eq:loss-split-explicit} yields
\[
\Loss_T(o)
\le
\max_{k\in[M]} B_k
+
\frac{1}{\beta}\log M
+
C^{\mix}(q_0;w_1)
-
\min_{o\in\mathcal O}\rho(o)\cdot q_0.
\]
This upper bound is independent of the time horizon $T$, the specific trade sequence $(r_t)$, and the realized outcome $o$. Taking the supremum over all such quantities yields the stated bound. 
\end{proof}

\expressivness* 
\begin{proof} 
Define the convex conjugate of the mixed potential as
\[R^{\mix}(x):=(C^{\mix})^*(x)
=
\sup_{q\in\mathbb R^d}\bigl\{q\cdot x - C^{\mix}(q;w)\bigr\}.
\]
Since $C^{\text{mix}}$ is finite and convex on $\mathbb{R}^d$, its conjugate $R^{\text{mix}}$ is proper, closed, and convex. Furthermore, we let $R_k(x) := C_{\eta_k}^*(x)$ denote the conjugate of the $k$-th expert. Under the scaled-family construction, all experts share the same conjugate domain: $\text{dom}(R_k) = \text{dom}(C^*) = H(\rho(\mathcal{O}))$, which we abbreviate as $H$. We recall that
\[R^{\mix}(x)
=
\sup_{q \in \mathbb{R}^d}
\{ q \cdot x - C^{\mix}(q; w) \},
\qquad
\dom(R^{\mix}) = \{x : R^{\mix}(x) < \infty\}.\]
 We first prove $\dom(R^{\mix}) = H$ by showing inclusion in both directions.
\paragraph{(i) $H \subseteq \dom(R^{\mix})$.}
Fix any $k \in [M]$. Since the log-sum-exp dominates each term, we have
\[
C^{\mix}(q; w)
\;\ge\;
C_k(q) + \frac{1}{\beta}\log w(k).
\]
Rearranging gives
\[
q \cdot x - C^{\mix}(q; w)
\;\le\;
q \cdot x - C_k(q) - \frac{1}{\beta}\log w(k).
\]
Taking the supremum over $q$ yields
\[
R^{\mix}(x)
\;\le\;
R_k(x) - \frac{1}{\beta}\log w(k).
\]
Fix any $x \in H = \dom(R_k)$. Then $R_k(x) < \infty$, and since $w(k) > 0$, the constant $-\frac{1}{\beta}\log w(k)$ is finite which makes the right-hand side is finite, so $R^{\mix}(x) < \infty$.
By definition of the domain, $x \in \dom(R^{\mix})$, proving $H \subseteq \dom(R^{\mix})$.

\paragraph{(ii) $\dom(R^{\mix}) \subseteq H$.} Suppose $R^{\mix}(x) < \infty$. Using the upper bound for the log-sum-exp function, for all $q \in \mathbb{R}^d$
$$C^{\text{mix}}(q; w) \le \frac{1}{\beta}\log\left(M \max_k e^{\beta C_k(q)}\right) = \max_k C_k(q) + \frac{1}{\beta}\log M.$$
Let $C_{\max}(q) := \max_k C_k(q)$. Substituting this upper bound into the definition of $R^{\text{mix}}$ gives a lower bound on the conjugate
$$R^{\text{mix}}(x) \ge \sup_q \{ q \cdot x - C_{\max}(q) \} - \frac{1}{\beta}\log M = C_{\max}^*(x) - \frac{1}{\beta}\log M.$$
Suppose $x \in \text{dom}(R^{\text{mix}})$, meaning $R^{\text{mix}}(x) < \infty$. The inequality above forces $C_{\max}^*(x) < \infty$. By the properties of convex conjugates, the domain of the conjugate of a finite maximum of convex functions is the convex hull of the union of their individual conjugate domains. Then we have
$$\text{dom}(C_{\max}^*) = \text{conv}\left(\bigcup_{k=1}^M \text{dom}(R_k)\right) = \text{conv}(H) = H.$$
Thus, $C_{\max}^*(x) < \infty \implies x \in H$, completing the proof that $\text{dom}(R^{\text{mix}}) = H$.

We now show subdifferentiability in the relative interior.
Fix any interior belief \(\mu\in\operatorname{ri}(H)\). Since \(\dom(R^{\mix})=H\), we have
$\mu\in\operatorname{ri}(\dom(R^{\mix}))$.
Because \(R^{\mix}\) is proper and convex, 
(\cite{rockafellar1997convex}, Thm 23.4) guarantees that the subdifferential is non-empty, that is, $\partial R^{\text{mix}}(\mu) \neq \emptyset$. Thus, we can choose any $q \in \partial R^{\text{mix}}(\mu)$.
By (\cite{rockafellar1997convex}, Thm 23.5 / Corr 23.5.1), the subgradients of a function and its conjugate are inverses of each other
\[
q\in \partial R^{\mix}(\mu)
\quad\Longleftrightarrow\quad
\mu\in \partial C^{\mix}(q;w).
\]
By Lemma \ref{lem:master-grad-hess-varcov}, \(C^{\mix}(\cdot;w)\) is everywhere differentiable on \(\mathbb R^d\).
Consequently, its subdifferential at any $q$ is the singleton containing only the gradient
\[\partial C^{\mix}(q;w)=\{\nabla C^{\mix}(q;w)\}.\]
Thus, we have
\[
\mu\in \partial C^{\mix}(q;w)
\quad\Longrightarrow\quad
\mu=\nabla C^{\mix}(q;w).
\]

Since \(\mu\in\operatorname{ri}(H)\) was arbitrary, we conclude that for every interior belief \(\mu\) there exists a market state \(q\) with
$\nabla C^{\mix}(q;w)=\mu$. This proves expressiveness.
\end{proof}


\positiveupside*
\begin{proof}
Fix a round \(t\), an outcome \(o \in \mathcal O\), and a feasible
direction \(v\). Consider a trade of the form \(r = s v\) for \(s>0\).
We observe that the adaptive payment rule can be written as
\[
\Pay_t^{\mix}(r)
=
C^{\mix}(q_{t-1}+r;w_t)
-
C^{\mix}(q_{t-1};w_t)
+
\Delta_t^{\dyn}(r)
+
\fee_t,
\]
where the \emph{dynamic distortion} is defined as
\[
\Delta_t^{\dyn}(r)
:=
C^{\mix}(q_{t-1}+r;w_{t+1})
-
C^{\mix}(q_{t-1}+r;w_t).
\]
This term captures the change in the potential evaluated at the
post-trade state due to updating the mixture weights. We now analyze the
realized profit in outcome \(o\in\mathcal O\),
\[
\rho(o)\cdot (s v)-\Pay_t^{\mix}(s v).
\]

By local smoothness of the frozen potential
\(q\mapsto C^{\mix}(q;w_t)\), there exists a constant \(L_t\ge0\) such
that for all sufficiently small \(s>0\),
\[
C^{\mix}(q_{t-1}+s v;w_t)
\le
C^{\mix}(q_{t-1};w_t)
+
p_t\cdot (s v)
+
\frac{L_t}{2}\|v\|^2 s^2,
\]
where $p_t
=
\nabla C^{\mix}(q_{t-1};w_t)$.

This follows from a first-order Taylor expansion of
\(C^{\mix}(\cdot;w_t)\) at \(q_{t-1}\), with smoothness controlling the
second-order remainder.

Substituting the payment decomposition and applying the smoothness bound
yields
\begin{align*}
\rho(o)\cdot (s v)-\Pay_t^{\mix}(s v)
&=
\rho(o)\cdot (s v)
-
\Big(
C^{\mix}(q_{t-1}+s v;w_t)
-
C^{\mix}(q_{t-1};w_t)
\Big)
\\
&\quad
-
\Delta_t^{\dyn}(s v)
-
\fee_t
\\
&\ge
\rho(o)\cdot (s v)
-
p_t\cdot (s v)
-
\frac{L_t}{2}\|v\|^2 s^2
\\
&\quad
-
\big|\Delta_t^{\dyn}(s v)\big|
-
\fee_t.
\end{align*}

By the assumed stability of updates, there exist nonnegative constants
\(F_{0,t},F_{1,t},F_{2,t}\) such that
\[
\big|\Delta_t^{\dyn}(s v)\big|
+
\fee_t
\le
F_{0,t}
+
F_{1,t}\|v\| s
+
F_{2,t}\|v\|^2 s^2.
\]
Substituting this into the profit bound gives
\[
\rho(o)\cdot (s v)-\Pay_t^{\mix}(s v)
\ge
-
F_{0,t}
+
s(\rho(o)-p_t)\cdot v
-
F_{1,t}\|v\| s
-
\Big(
\frac{L_t}{2}
+
F_{2,t}
\Big)\|v\|^2 s^2.
\]

Define
\[
A
:=
(\rho(o)-p_t)\cdot v
-
F_{1,t}\|v\|, \quad B
:=
\Big(
\frac{L_t}{2}
+
F_{2,t}
\Big)\|v\|^2.
\]
Then the realized profit satisfies
\[
\rho(o)\cdot (s v)-\Pay_t^{\mix}(s v)
\ge
-
F_{0,t}
+
sA
-
s^2B.
\]

If \(A>0\) and
$A^2
>
4BF_{0,t}$,
then the quadratic $-
F_{0,t}
+
sA
-
s^2B$
has two distinct positive roots
$0<s_{\min}<s_{\max}$,
and is strictly positive on the interval
$s\in(s_{\min},s_{\max})$.
Therefore,
\[
\rho(o)\cdot (s v)-\Pay_t^{\mix}(s v) > 0
\qquad
\forall\, s\in(s_{\min},s_{\max}),
\]
which proves positive directional upside.
\end{proof}



\section{Proofs and Additional Results from Section 5}
\label{app:slippage-liability}
For each expert $k$ and round $t$, define the \emph{slippage} (Bregman divergence) for trade $r_t = q_t - q_{t-1}$ as
\[
S_{k,t}
:=
D_{C_k}(q_t,q_{t-1})
=
C_k(q_t)-C_k(q_{t-1})-\nabla C_k(q_{t-1})\cdot r_t,
\]
and the liability as
\[
L_k(q):=\max_{o}\rho(o)\cdot q -\bigl( C_k(q) - C_k(0)\bigr).
\]

The hybrid signal is
\[
\Gamma^{\hyb}_{k,t}
=
a\Big(S_{k,t}-\frac{1}{M}\sum_{j=1}^M S_{j,t}\Big)
+
b L_k(q_t),
\quad a,b>0.
\]

Define the surrogate signal
\[
\Gamma_t^{\surr}
:=
\sum_{k=1}^M w_t(k)\Gamma^{\hyb}_{k,t}.
\]

Define the mixed signal
\[
\Gamma^{\hyb}_{\mix,t}
=
a\Big(
D_{C_{\mix,t}}(q_t,q_{t-1})
-
\frac{1}{M}\sum_{j=1}^M S_{j,t}
\Big)
+
b L_{\mix,t}(q_t)
.
\]

At round \(t\), define the mixed cost
\[
C_{\mix,t}(q)
:=
C^{\mix}(q;w_t)
=
\frac{1}{\beta}
\log\!\Big(\sum_{k=1}^M w_t(k)\,e^{\beta C_k(q)}\Big).
\]
The mixed slippage is the Bregman divergence of \(C_{\mix,t}\)
\[
D_{C_{\mix,t}}(q_t,q_{t-1})
=
C_{\mix,t}(q_t)
-
C_{\mix,t}(q_{t-1})
-
\nabla C_{\mix,t}(q_{t-1})\cdot (q_t-q_{t-1}),\]
where the mixed gradient is the posterior-weighted sum of expert gradients
\[
\nabla C_{\mix,t}(q)
=
\sum_{k=1}^M \pi_t(k;q)\,\nabla C_k(q),
\quad
\pi_t(k;q)
=
\frac{w_t(k)e^{\beta C_k(q)}}{\sum_{j=1}^M w_t(j)e^{\beta C_j(q)}}.
\]

The mixed liability is
\[
L_{\mix,t}(q)
:=
\max_{o}
\rho(o)\cdot q - \bigl(C_{\mix,t}(q) - C_{\mix,t}(0) \bigr).
\]

We evaluate all quantities with the \emph{frozen} mixture \(C_{\mix,t}\), that is, the learning weights \(w_t\) are held fixed within round \(t\).

\begin{lemma}[Slippage bound]
\label{lem:hybrid-slippage}
For every round \(t\),
\begin{align*}
D_{C_{\mix,t}}(q_t,q_{t-1})
&\le
\sum_{k=1}^M \pi_t(k;q_t)\, S_{k,t} 
+
\sum_{k=1}^M
(\pi_t(k;q_t)-\pi_t(k;q_{t-1}))
\,\nabla C_k(q_{t-1})\cdot r_t.
\end{align*}
\end{lemma}

\begin{proof}
By definition of the Bregman divergence for the mixed potential
\[D_{C_{\mix,t}}(q_t,q_{t-1})
=C_{\mix,t}(q_t)-C_{\mix,t}(q_{t-1})
-\nabla C_{\mix,t}(q_{t-1})\cdot r_t.\]
Substituting the mixed gradient $\nabla C_{\mix,t}(q)=\sum_{k=1}^M \pi_t(k;q)\,\nabla C_k(q)$ yields
\[D_{C_{\mix,t}}(q_t,q_{t-1})
=C_{\mix,t}(q_t)-C_{\mix,t}(q_{t-1})
-\sum_{k=1}^M \pi_t(k;q_{t-1})\,\nabla C_k(q_{t-1})\cdot r_t.\]
We now bound the difference in the cost functions. Recall that $C_{\text{mix},t}$ is constructed using the log-sum-exp function on expert costs. Because the log-sum-exp function is globally convex with respect to its inputs and its gradient is the softmax distribution $\pi_t(k; q)$, applying the first-order convexity condition $F(y) - F(x) \le \nabla F(y) \cdot (y - x)$ yields
$$C_{\text{mix},t}(q_t) - C_{\text{mix},t}(q_{t-1}) \le \sum_{k=1}^M \pi_t(k; q_t) \left( C_k(q_t) - C_k(q_{t-1}) \right).$$
Substituting this upper bound into our divergence expression
$$D_{C_{\text{mix},t}}(q_t, q_{t-1}) \le \sum_{k=1}^M \pi_t(k; q_t) \left( C_k(q_t) - C_k(q_{t-1}) \right) - \sum_{k=1}^M \pi_t(k; q_{t-1}) \nabla C_k(q_{t-1}) \cdot r_t.$$
We add and subtract the term $\sum_{k=1}^M \pi_t(k; q_t) \nabla C_k(q_{t-1}) \cdot r_t$ to align the inner product with the posterior weights at $q_t$:
\begin{align*}D_{C_{\text{mix},t}}(q_t, q_{t-1}) \le \sum_{k=1}^M \pi_t(k; q_t) \left[ C_k(q_t) - C_k(q_{t-1}) - \nabla C_k(q_{t-1}) \cdot r_t \right] & \\\quad + \sum_{k=1}^M \left( \pi_t(k; q_t) - \pi_t(k; q_{t-1}) \right) \nabla C_k(q_{t-1}) \cdot r_t.\end{align*}
Recognizing the bracketed term as the expert slippage $S_{k,t} = C_k(q_t) - C_k(q_{t-1}) - \nabla C_k(q_{t-1}) \cdot r_t$, the result immediately follows.
\end{proof}


\begin{lemma}[Liability-level bound]
\label{lem:hybrid-liability}
Assume without loss of generality that the candidate cost functions are normalized such that the initial state cost is zero, i.e., $C_k(0) = 0$ for all $k$. For every round \(t\) and state \(q_t\),
$
L_{\mix,t}(q_t)
\le
\sum_{k=1}^M w_t(k)\,L_k(q_t)$.
\end{lemma}

\begin{proof}

We prove the claim in two steps.

We first show a lower bound on the mixed cost increment.
Consider the function
\[
F(x_1,\dots,x_M)
:=
\frac{1}{\beta}
\log\!\Big(\sum_{k=1}^M w_t(k)e^{\beta x_k}\Big).
\]
This is the weighted log-sum-exp map, so it is convex. By Jensen's inequality applied to the convex function \(u\mapsto e^{\beta u}\), we have
\[
\sum_{k=1}^M w_t(k)e^{\beta x_k}
\ge
\exp\!\Big(\beta \sum_{k=1}^M w_t(k)x_k\Big).
\]
Taking \(\frac{1}{\beta}\log\) of both sides yields
\[
F(x_1,\dots,x_M)
\ge
\sum_{k=1}^M w_t(k)x_k.
\]

Now apply this with \(x_k=C_k(q_t)\). Then
\[
C_{\mix,t}(q_t)
=
\frac{1}{\beta}
\log\!\Big(\sum_{k=1}^M w_t(k)e^{\beta C_k(q_t)}\Big)
\ge
\sum_{k=1}^M w_t(k)C_k(q_t).
\]
Applying the same inequality at \(q=0\) gives
\[
C_{\mix,t}(0)
=
\frac{1}{\beta}
\log\!\Big(\sum_{k=1}^M w_t(k)e^{\beta C_k(0)}\Big)
\ge
\sum_{k=1}^M w_t(k)C_k(0).\]
By our assumption, $C_k(0) = 0$ for all $k$, which implies $C_{\text{mix},t}(0) = \frac{1}{\beta}\log(\sum_{k=1}^M w_t(k)e^0) = 0$. Subtracting these zero-valued terms from their respective sides preserves the inequality:
\[
C_{\mix,t}(q_t)-C_{\mix,t}(0)
\ge
\sum_{k=1}^M w_t(k)\bigl(C_k(q_t)-C_k(0)\bigr).
\]

Now, we show the relation to liability.
Fix any outcome \(o \in \mathcal{O}\). Multiplying the inequality from the previous part by $-1$ and adding $\rho(o) \cdot q_t$ to both sides gives
\begin{align*}
\rho(o)\cdot q_t - \bigl(C_{\mix,t}(q_t)-C_{\mix,t}(0)\bigr)
&\le
\rho(o)\cdot q_t
-
\sum_{k=1}^M w_t(k)\bigl(C_k(q_t)-C_k(0)\bigr) \\
&=
\sum_{k=1}^M w_t(k)
\Bigl(
\rho(o)\cdot q_t - \bigl(C_k(q_t)-C_k(0)\bigr)
\Bigr).
\end{align*}
where the equality holds because the weights $w_t(k)$ sum to one. Since this bound holds for every outcome, taking the maximum over $o \in \mathcal O$ on the left side yields the definition of the mixed liability
\[
L_{\mix,t}(q_t)
=
\max_{o}
\Bigl\{
\rho(o)\cdot q_t - \bigl(C_{\mix,t}(q_t)-C_{\mix,t}(0)\bigr)
\Bigr\}
\]
\[
\le
\max_{o}
\sum_{k=1}^M w_t(k)
\Bigl(
\rho(o)\cdot q_t - \bigl(C_k(q_t)-C_k(0)\bigr)
\Bigr).
\]
We now use the elementary inequality $\max_{o}\sum_{k=1}^M w_t(k)f_k(o)
\le
\sum_{k=1}^M w_t(k)\max_{o} f_k(o)$, which holds true for any non-negative weights \(w_t(k)\). Setting $f_k(o) = \rho(o) \cdot q_t - (C_k(q_t) - C_k(0))$, we get:
\begin{align*}
L_{\mix,t}(q_t)
&\le
\sum_{k=1}^M w_t(k)
\max_{o}
\Bigl\{
\rho(o)\cdot q_t - \bigl(C_k(q_t)-C_k(0)\bigr)
\Bigr\} =
\sum_{k=1}^M w_t(k)\,L_k(q_t).
\end{align*}
\end{proof}


\begin{theorem}[Hybrid bridge inequality]
\label{thm:hybrid-bridge-slippage}
For each round $t$,
\[
\Gamma^{\hyb}_{\mix,t}
\le
\Gamma_t^{\surr}
+
a E^{\slip,1}_t
+
a E^{\slip,2}_t,
\]
where
\[
E^{\slip,1}_t
=
\sum_{k=1}^M (\pi_t(k;q_t)-w_t(k)) S_{k,t},
\]
\[
E^{\slip,2}_t
=
\sum_{k=1}^M (\pi_t(k;q_t)-\pi_t(k;q_{t-1}))
\,\nabla C_k(q_{t-1})\cdot r_t.
\]
\end{theorem}

\begin{proof}
By combining the slippage bound (Lemma~\ref{lem:hybrid-slippage}) and the liability-level bound (Lemma~\ref{lem:hybrid-liability}) with the definition of the mixed hybrid signal, we get
\begin{align*}
\Gamma^{\hyb}_{\mix,t}
&\le
a\Big(
\sum_{k=1}^M \pi_t(k;q_t) S_{k,t}
-
\frac{1}{M}\sum_j S_{j,t}
\Big)
+
b\sum_{k=1}^M w_t(k)L_k(q_t)
+
a E^{\slip,2}_t.
\end{align*}
Next, recall the term expansion of the surrogate signal
\[
\Gamma_t^{\surr}
=
a\sum_{k=1}^M w_t(k) S_{k,t}
-
\frac{a}{M}\sum_j S_{j,t}
+
b\sum_{k=1}^M w_t(k)L_k(q_t).
\]
We note that the baseline penalty term ($- \frac{a}{M} \sum_{j=1}^M S_{j,t}$) and the liability term ($b \sum_{k=1}^M w_t(k) L_k(q_t)$) are identical in both expressions. By subtracting $\Gamma_t^{\text{surr}}$ from the right-hand side of the upper bound, these terms cancel perfectly, yielding
\[
\Gamma^{\hyb}_{\mix,t}
\le
\Gamma_t^{\surr}
+
a\sum_{k=1}^M (\pi_t(k;q_t)-w_t(k)) S_{k,t}
+
a E^{\slip,2}_t,
\]
Recognizing the first term on the right-hand side as $a E_t^{\text{slip},1}$ and rearranging proves the claim.
\end{proof}

\begin{theorem}[Tracking guarantee for the surrogate hybrid objective]
\label{thm:hybrid-surrogate}
Fix a horizon $T \ge 1$ and a switch budget $J \ge 0$.
Assume the weight sequence $(w_t)_{t=1}^T$ is generated by an online
learning algorithm that satisfies the following \emph{tracking-regret}
property:

for every loss array $\ell_{k,t}\in[0,U]$,
\[
\sum_{t=1}^T \sum_{k=1}^M w_t(k)\ell_{k,t}
-
\min_{j_{1:T}:S(j_{1:T})\le J}
\sum_{t=1}^T \ell_{j_t,t}
\le
U\sqrt{2T\,A_{J,M,T}},
\]
where
\[
A_{J,M,T}
:=
(J+1)\log M
+
J\log\!\Bigl(\frac{e(T-1)}{J}\Bigr),
\]
with the convention $0\log(c/0):=0$.

Suppose additionally that the hybrid expert signals are uniformly bounded:
\[
|\Gamma^{\hyb}_{k,t}| \le B
\qquad
\forall k\in[M],\ \forall t\in[T].
\]

Then the surrogate hybrid signal satisfies
\[
\sum_{t=1}^T \Gamma^{\surr}_t
-
\min_{j_{1:T}:S(j_{1:T})\le J}
\sum_{t=1}^T \Gamma^{\hyb}_{j_t,t}
\le
2B\sqrt{2T\,A_{J,M,T}}.
\]
\end{theorem}

\begin{proof}
We reduce the statement to the tracking-regret guarantee by shifting the hybrid signals into the nonnegative range. We define
\[
\ell_{k,t}:=\Gamma_{k,t}^{\hyb}+B.
\]
By assumption (i), we have
$$-B\le \Gamma_{k,t}^{\hyb}\le B
\quad \forall\,k,t$$ 
and it follows that 
$$0\le \ell_{k,t}\le 2B.$$ 
Therefore, assumption (ii) applies with \(U=2B\), yielding
\[
\sum_{t=1}^T\sum_{k=1}^M w_t(k)\ell_{k,t}
-
\min_{j_{1:T}:S(j_{1:T})\le J}
\sum_{t=1}^T \ell_{j_t,t}
\le
2B\sqrt{2T\,A_{J,M,T}}.
\tag{1}
\]
We now expand the two terms in \((1)\). First, since \(w_t\in\Delta^M\), we have \(\sum_{k=1}^M w_t(k)=1\) for every \(t\). Thus, it follows that
\[
\sum_{t=1}^T\sum_{k=1}^M w_t(k)\ell_{k,t}
=
\sum_{t=1}^T\sum_{k=1}^M w_t(k)\bigl(\Gamma_{k,t}^{\hyb}+B\bigr)\]
\[
=
\sum_{t=1}^T\sum_{k=1}^M w_t(k)\Gamma_{k,t}^{\hyb}
+
B\sum_{t=1}^T\sum_{k=1}^M w_t(k)
=
\sum_{t=1}^T \Gamma_t^{\surr}+BT.
\tag{2}
\]

Second, for any comparator sequence \(j_{1:T}\), we have
\[
\sum_{t=1}^T \ell_{j_t,t}
=
\sum_{t=1}^T \bigl(\Gamma_{j_t,t}^{\hyb}+B\bigr)
=
\sum_{t=1}^T \Gamma_{j_t,t}^{\hyb}+BT.
\]
Since the additive term \(BT\) does not depend on the comparator sequence, it
passes unchanged through the minimum:
\[
\min_{j_{1:T}:S(j_{1:T})\le J}
\sum_{t=1}^T \ell_{j_t,t}
=
\min_{j_{1:T}:S(j_{1:T})\le J}
\left(\sum_{t=1}^T \Gamma_{j_t,t}^{\hyb}+BT\right)
\]
\[
=
\min_{j_{1:T}:S(j_{1:T})\le J}
\sum_{t=1}^T \Gamma_{j_t,t}^{\hyb}
+
BT.
\tag{3}
\]

Substituting \((2)\) and \((3)\) into \((1)\), we obtain
\[
\left(\sum_{t=1}^T \Gamma_t^{\surr}+BT\right)
-
\left(
\min_{j_{1:T}:S(j_{1:T})\le J}
\sum_{t=1}^T \Gamma_{j_t,t}^{\hyb}
+
BT
\right)
\le
2B\sqrt{2T\,A_{J,M,T}}.
\]
The common shift \(BT\) cancels, leaving the desired bound
\[
\sum_{t=1}^T \Gamma_t^{\surr}
-
\min_{j_{1:T}:S(j_{1:T})\le J}
\sum_{t=1}^T \Gamma_{j_t,t}^{\hyb}
\le
2B\sqrt{2T\,A_{J,M,T}}.
\]
\end{proof}

To simplify notation for the actual (fee-adjusted) payment rule, we define the generic switch-corrected potential for any state $q \in \mathbb{R}^d$, weight vector $w \in \Delta_M$, and scalar $z \in \mathbb{R}$ as
\[\Phi(q, w, z) := C_{\text{mix}}(q; w) + z.\]
Let $\fee_s$ denote the weight-update fee incurred at round $s$. We define the cumulative adaptation fee up to round $t$ as
$$z_t := \sum_{s=1}^t \fee_s, \quad z_0 = 0.$$
Evaluating at the post-update market state and cumulative fee for round $t$ yields
$$\Phi(q_t, w_{t+1}, z_t) = C_{\text{mix}}(q_t; w_{t+1}) + z_t.$$
The realized payment rule $\text{Pay}_t^{\text{mix}}$ for the round naturally telescopes
$$\text{Pay}_t^{\text{mix}} = \Phi(q_t, w_{t+1}, z_t) - \Phi(q_{t-1}, w_t, z_{t-1}).$$
Expanding this difference yields exactly the standard cost increment plus the adaptation fee $$C_{\text{mix}}(q_t; w_{t+1}) - C_{\text{mix}}(q_{t-1}; w_t) + \fee_t.$$ 
Thus, $\Phi$ absorbs the accumulated friction, allowing us to express the realized market dynamics as a strict potential difference despite the dynamically changing weights.
\begin{theorem}[Transfer to the actual hybrid signal] 
\label{thm:transfer-clean}
Fix a horizon \(T\) and a comparator sequence \(j_{1:T}\in[M]^T\). 
For each round \(t\), let \(w_t,w_{t+1}\in\Delta^M\), and define
\[
\pi_t(k;q_t)
=
\frac{w_t(k)e^{\beta C_k(q_t)}}
{\sum_{m=1}^M w_t(m)e^{\beta C_m(q_t)}}.
\]
Using the surrogate signal, mixed signal, actual hybrid signal, and the error terms $E_t^{\text{slip},2}$ as defined in Theorem \ref{thm:hybrid-bridge-slippage}, we have:
\begin{align*}\sum_{t=1}^T \Gamma_t^{\Phi, \text{hyb}} - \sum_{t=1}^T \Gamma_{j_t,t}^{\text{hyb}} &\le R_T^{\text{surr}} + a \sum_{t=1}^T |(\pi_t - w_t) \cdot S_t| \\&\quad + a \sum_{t=1}^T \left[ \nabla_w C_{\text{mix}}(q_t; w_t) \cdot (w_{t+1} - w_t) \right] \\
&\quad + a \sum_{t=1}^T E_t^{\text{slip},2} + a \sum_{t=1}^T \fee_t,\end{align*}
where
$$R_T^{\text{surr}} := \sum_{t=1}^T \Gamma_t^{\text{surr}} - \sum_{t=1}^T \Gamma_{j_t, t}^{\text{hyb}}, \quad S_t := (S_{1,t}, \dots, S_{M,t}).$$
\end{theorem}

\begin{proof} By the definition of the actual payment rule with the adaptation fee, we have:$$\text{Pay}_t^{\text{mix}}(r_t) = C_{\text{mix}}(q_t; w_{t+1}) - C_{\text{mix}}(q_{t-1}; w_t) + \fee_t.$$
Because the fee $\fee_t$ strictly modifies the payment (and thus the realized slippage), it inherits the slippage scaling factor $a$. Accordingly, the actual hybrid signal decomposes as:$$\Gamma_t^{\Phi, \text{hyb}} = \Gamma_{\text{mix},t}^{\text{hyb}} + a(C_{\text{mix}}(q_t; w_{t+1}) - C_{\text{mix}}(q_t; w_t)) + a \fee_t.$$
By the weight-update bound (cf. Corollary \ref{corr:costupdatedbound}),$$C_{\text{mix}}(q_t; w_{t+1}) - C_{\text{mix}}(q_t; w_t) \le \left[ \nabla_w C_{\text{mix}}(q_t; w_t) \cdot (w_{t+1} - w_t) \right]_+,$$and hence$$\Gamma_t^{\Phi, \text{hyb}} \le \Gamma_{\text{mix},t}^{\text{hyb}} + a \left[ \nabla_w C_{\text{mix}}(q_t; w_t) \cdot (w_{t+1} - w_t) \right]_+ + a \fee_t.$$
Next, we apply the mixed-to-surrogate bridge inequality (Theorem \ref{thm:hybrid-bridge-slippage}) to the $\Gamma_{\text{mix},t}^{\text{hyb}}$ term:$$\Gamma_{\text{mix},t}^{\text{hyb}} \le \Gamma_t^{\text{surr}} + a (\pi_t - w_t) \cdot S_t + a E_t^{\text{slip},2}.$$
Taking absolute values on the mismatch term preserves the inequality:
$$\Gamma_{\text{mix},t}^{\text{hyb}} \le \Gamma_t^{\text{surr}} + a |(\pi_t - w_t) \cdot S_t| + a E_t^{\text{slip},2}.$$
Substituting this into the previous display gives:
\begin{align*}\Gamma_t^{\Phi, \text{hyb}} &\le \Gamma_t^{\text{surr}} + a |(\pi_t - w_t) \cdot S_t| + a E_t^{\text{slip},2} + a \left[ \nabla_w C_{\text{mix}}(q_t; w_t) \cdot (w_{t+1} - w_t) \right]_+ + a \fee_t.
\end{align*}
Summing over $t = 1, \dots, T$ and subtracting the comparator sequence $\sum_{t=1}^T \Gamma_{j_t,t}^{\text{hyb}}$ from both sides yields the claim.

\end{proof}

\begin{theorem}[Tracking guarantee for the realized adaptive market]
\label{thm:actual-hybrid-regret}
Fix a horizon $T$ and a comparator sequence
$j_{1:T}$ satisfying $S(j_{1:T}) \le J$.

Assume the following.

\begin{enumerate}
\item \textbf{Bounded hybrid signals.}
For all $k\in[M]$ and $t\in[T]$, $|\Gamma^{\hyb}_{k,t}| \le B.$
\item \textbf{Tracking guarantee for the surrogate objective.}
The weight sequence $(w_t)_{t=1}^T$ is generated by an online learning
algorithm satisfying the tracking-regret guarantee of
Theorem~\ref{thm:hybrid-surrogate}.

\item \textbf{Controlled updates and fees.}
There exist constants $C_w,C_f\ge0$ such that
\[
a\sum_{t=1}^T
\bigl[
\nabla_w C_{\mix}(q_t;w_t)\cdot(w_{t+1}-w_t)
\bigr]_+
\le
C_w,
\]
and
\[
a\sum_{t=1}^T \fee_t \le C_f.
\]
\end{enumerate}

Then the realized adaptive market satisfies
\[
\sum_{t=1}^T \Gamma_t^{\Phi,\hyb}
-
\sum_{t=1}^T \Gamma^{\hyb}_{j_t,t}
\le
2B\sqrt{2T\,A_{J,M,T}}
+
a\sum_{t=1}^T |(\pi_t-w_t)\cdot S_t|
+
a\sum_{t=1}^T E_t^{\slip,2}
+
C_w
+
C_f.
\]
\end{theorem}

\begin{proof}
We apply the exact transfer bound in Theorem~\ref{thm:transfer-clean}, which decomposes the realized hybrid regret as:
\begin{align*}\sum_{t=1}^T \Gamma_t^{\Phi, \text{hyb}} - \sum_{t=1}^T \Gamma_{j_t,t}^{\text{hyb}} &\le R_T^{\text{surr}} + a \sum_{t=1}^T |(\pi_t - w_t) \cdot S_t| + a \sum_{t=1}^T E_t^{\text{slip},2} \\ & + a \sum_{t=1}^T \left[ \nabla_w C_{\text{mix}}(q_t; w_t) \cdot (w_{t+1} - w_t) \right]_+ + a \sum_{t=1}^T \fee_t.\end{align*}
By Assumption (ii), the surrogate tracking regret is bounded by Theorem~\ref{thm:hybrid-surrogate} as $R_T^{\text{surr}} \le 2B \sqrt{2T A_{J,M,T}}$. By Assumption (iii), the cumulative weight-update costs and adaptation fees are bounded by $C_w$ and $C_f$, respectively. Substituting these three bounds directly into the transfer inequality yields the claim.
\end{proof}

\begin{lemma}[Dominance implies small mismatch]
\label{≈}
For any \(u,v\in\Delta^M\) and any \(j\in[M]\),
\[
\|u-v\|_1
\le
2(1-u(j)) + 2(1-v(j)).
\]
\end{lemma}

\begin{proof}
We decompose the $L_1$ distance into the distinguished coordinate $j$ and the remaining coordinates:
\[\|u-v\|_1
=|u(j)-v(j)|+
\sum_{k\neq j}|u(k)-v(k)|.\]
For the distinguished coordinate, we apply the triangle inequality through $1$. Since $u(j), v(j) \le 1$, we have
\[
|u(j)-v(j)|
\le
(1-u(j))+(1-v(j)),
\]
For the remaining coordinates, because all elements of the simplex are non-negative, the absolute difference is strictly bounded by the sum:
\[
\sum_{k\neq j}|u(k)-v(k)|
\le
\sum_{k\neq j}(u(k)+v(k))
=
(1-u(j))+(1-v(j)).
\]
Because $u,v\in\Delta^M$, their components sum to $1$. Thus, $\sum_{k \neq j} u(k) = 1 - u(j)$ and $\sum_{k \neq j} v(k) = 1 - v(j)$, yielding:
\[
\|u-v\|_1
\le
2(1-u(j)) + 2(1-v(j)).
\]
Adding the bounds for the distinguished coordinate and the remaining components proves the claim.
\end{proof}


\begin{corollary}[Dominant-regime regret] 
\label{cor:dominant-clean}

Under the assumptions of Theorem~\ref{thm:actual-hybrid-regret}, suppose additionally:

\begin{itemize}
\item[(i)] \textbf{Bounded trades and gradients}: For all $k$ and $t$
$$\|r_t\| \le R, \quad \|\nabla C_k(q)\|_* \le G$$

\item[(ii)] \textbf{Bounded drift and scaling penalties}: The total intra-trade drift is bounded by 
$$a \sum_{t=1}^T E_t^{\text{slip},2} \le C_{\text{drift}}.$$
Furthermore, the adaptation penalties scale with the number of switches, such that 
$$C_w + C_f = \mathcal{O}(|B|).$$ 

\item[(iii)] \textbf{(Dominance outside burn-in)}
There exists \(B \subseteq [T]\) with \(|B|\le (J+1)B_0\) such that for \(t \notin B\),
\[
w_t(j_t) \ge 1 - \varepsilon^-_t,
\qquad
\pi_t(j_t;q_t) \ge 1 - \varepsilon^\pi_t,
\]
and the sum of these errors outside the burn-in period is bounded
\[
\sum_{t\notin B} (\varepsilon^-_t + \varepsilon^\pi_t) \le C_{\dom}.
\]
\end{itemize}

Then the realized hybrid regret satisfies
\[
\sum_{t=1}^T \Gamma^{\Phi,\hyb}_t
-
\sum_{t=1}^T \Gamma^{\hyb}_{j_t,t}
\;\le\;
2B\sqrt{2T A_{J,M,T}}
+
4aGR\, C_{\dom}
+
\mathcal{O}(|B|).\]

In particular, if $J$, $|B|$, $C_{\text{dom}}$, and $C_{\text{drift}}$ are fixed with respect to $T$, the average regret vanishes
\[\frac{1}{T}
\Big(
\sum_{t=1}^T \Gamma^{\Phi,\hyb}_t
-
\sum_{t=1}^T \Gamma^{\hyb}_{j_t,t}
\Big)
\to 0.\]
\end{corollary}

\begin{proof}
Under the bounded-gradient assumption (i), we obtain a uniform bound on expert slippage. By definition
\[S_{k,t}=D_{C_k}(q_t,q_{t-1})
=C_k(q_t)-C_k(q_{t-1})-\nabla C_k(q_{t-1})\cdot r_t.\]
Applying the Mean Value Theorem, followed by Hölder's and the triangle inequality, gives
$$S_{k,t} = C_k(q_t) - C_k(q_{t-1}) - \nabla C_k(q_{t-1}) \cdot r_t \le 2 \sup_{q} \|\nabla C_k(q)\|_* \|r_t\| \le 2GR.$$
Thus, the vector of expert slippages is bounded by $\|S_t\|_\infty \le 2GR$. By Hölder's inequality, the mismatch term satisfies
$$|(\pi_t - w_t) \cdot S_t| \le \|S_t\|_\infty \|\pi_t - w_t\|_1 \le 2GR \|\pi_t - w_t\|_1.$$
By Lemma~\ref{≈}, evaluating the $L_1$ distance at the optimal coordinate $j_t$ gives
$$\|\pi_t - w_t\|_1 \le 2(1 - \pi_t(j_t; q_t)) + 2(1 - w_t(j_t)).$$
Summing over $t = 1, \dots, T$ and separating the burn-in and non-burn-in rounds yields
$$\sum_{t=1}^T \|\pi_t - w_t\|_1 = \sum_{t \notin B} \|\pi_t - w_t\|_1 + \sum_{t \in B} \|\pi_t - w_t\|_1 \le 2C_{\text{dom}} + 2|B|,$$
where we used Assumption (iii) for $t \notin B$, and the  bound $\|\pi_t - w_t\|_1 \le 2$ for $t \in B$. Multiplying by $a\|S_t\|_\infty$, the cumulative mismatch penalty is bounded by
$$a \sum_{t=1}^T |(\pi_t - w_t) \cdot S_t| \le 4aGR C_{\text{dom}} + 4aGR|B|.$$
Finally, we substitute this mismatch bound into Theorem~\ref{thm:actual-hybrid-regret}. By Assumption (ii), the drift term is bounded by $C_{\text{drift}}$, and the update penalties $C_w + C_f$ scale as $\mathcal{O}(|B|)$. Since $4aGR|B|$ is also $\mathcal{O}(|B|)$, these terms combine, yielding the result.
\end{proof}


\section{Positive Directional Upside for LMSR Mixtures}
\label{app:positive-upside-general}

This section proves Theorem~\ref{thm:positive-upside} and verifies that
its assumptions hold for LMSR mixtures under stable weight updates. We
first establish the result for \(N\) outcomes and \(M\) liquidity
levels, and then specialize to the two-liquidity case used in the
simulations.

Let the expert cost functions be
\[
C_k(q)
=
b_k \log\!\Big(\sum_{i=1}^N e^{q_i/b_k}\Big),
\qquad
k=1,\dots,M,
\]
and define the mixed cost
\[
C^{\mix}(q;w)
=
\frac1\beta
\log\!\Big(
\sum_{k=1}^M
w(k)e^{\beta C_k(q)}
\Big).
\]
Let $p_t
:=
\nabla C^{\mix}(q_{t-1};w_t)$,
and consider a directional trade
$r=s v$.
The adaptive payment rule is
\[
\Pay_t^{\mix}(r)
=
C^{\mix}(q_{t-1}+r;w_t)
-
C^{\mix}(q_{t-1};w_t)
+
\Delta_t^{\dyn}(r)
+
\fee_t.
\]

For the \(N\)-outcome LMSR, the Hessian satisfies the uniform bound
\[
\nabla^2 C_k(q)
\preceq
\frac{1}{2b_k}I.
\]
Hence, by
Theorem~\ref{thm:smoothness-logsum-master},
\[
\nabla^2 C^{\mix}(q;w_t)
\preceq
\Big(
\frac{1}{2b_{\min}}
+
\beta G^2
\Big)I,
\]
where
\[
b_{\min}:=\min_k b_k,
\]
and \(G\) is a uniform bound on the expert gradients over the bounded
reachable region.

Therefore,
\[
C^{\mix}(q_{t-1}+r;w_t)
\le
C^{\mix}(q_{t-1};w_t)
+
p_t\cdot r
+
\frac{L_t}{2}\|r\|^2,
\]
with
\[
L_t
\le
\frac{1}{2b_{\min}}
+
\beta G^2.
\]

Weights are updated via fixed-share using the hybrid signal
\[
\Gamma^{\hyb}_{k,t}
=
a\Big(
S_{k,t}
-
\frac1M\sum_{j=1}^M S_{j,t}
\Big)
+
b\,L_k(q_t),
\qquad
a,b>0,
\]
where
$S_{k,t}
=
D_{C_k}(q_t,q_{t-1})$.

Assume trades remain in a bounded reachable region. Then both slippage
and liability are bounded, and hence the hybrid signal is uniformly
bounded. Under fixed-share updates with learning rate \(\eta\)
and share parameter \(\alpha\),
\[
\|w_{t+1}-w_t\|
=
O(\eta)+O(\alpha).
\]

Assume additionally that the weights remain uniformly interior:
\[
w_t(k)\ge\delta>0,
\qquad
\forall k,t.
\]
By Lemma~\ref{lem:cmix-concave-w},
\(C^{\mix}(q;w)\) is uniformly Lipschitz in \(w\), with Lipschitz
constant independent of \(q\). Consequently, both the dynamic distortion
and the switching fee satisfy
\[
|\Delta_t^{\dyn}(sv)|
+
\fee_t
\le
F_{0,t}
+
F_{1,t}\|v\| s
+
F_{2,t}\|v\|^2 s^2,
\]
where
\[
F_{0,t}=O(\eta+\alpha),
\qquad
F_{1,t},F_{2,t}=O(\eta+\alpha).
\]

\begin{corollary}[Positive directional upside for LMSR mixtures]
\label{cor:lmsr-positive-upside-general}
Fix a round \(t\), and suppose trades remain in a bounded reachable
region. Assume there exists \(\epsilon>0\) such that
\(
(\rho(o)-p_t)\cdot v \ge \epsilon\|v\|.
\)
If the learning parameters satisfy
\(
\eta+\alpha \le c\,\epsilon^2
\)
for a sufficiently small constant \(c>0\), then there exists a
nontrivial interval \((s_{\min},s_{\max})\) such that every trade
\(sv\) with \(s\in(s_{\min},s_{\max})\) yields strictly positive
realized profit.
\end{corollary}

\begin{proof}
From the proof of Theorem~\ref{thm:positive-upside}, the realized profit
satisfies
\[
\rho(o)\cdot(sv)-\Pay_t^{\mix}(sv)
\ge
-
F_{0,t}
+
s(\rho(o)-p_t)\cdot v
-
F_{1,t}\|v\|s
-
\Bigl(
\frac{L_t}{2}+F_{2,t}
\Bigr)\|v\|^2 s^2.
\]
Define
\(
A_t := (\rho(o)-p_t)\cdot v - F_{1,t}\|v\|
\)
and
\(
Q_t := \bigl(\frac{L_t}{2}+F_{2,t}\bigr)\|v\|^2.
\)
Then
\[
\rho(o)\cdot(sv)-\Pay_t^{\mix}(sv)
\ge
-
F_{0,t}
+
A_t s
-
Q_t s^2.
\]

Since
\(
F_{0,t},F_{1,t},F_{2,t}=O(\eta+\alpha),
\)
choosing
\(
\eta+\alpha \le c\,\epsilon^2
\)
for sufficiently small \(c>0\) ensures that
\(
A_t>0
\)
and
\(
A_t^2>4Q_tF_{0,t}.
\)
Hence the quadratic
\(
-F_{0,t}+A_t s-Q_t s^2
\)
is strictly positive on a nontrivial interval
\(
s\in(s_{\min},s_{\max}),
\)
which yields positive directional upside in the sense of
Definition~\ref{def:positive-upside}.
\end{proof}

We now specialize to the two-liquidity case used in the simulations,
with
\[
b_1=1,
\qquad
b_2=12,
\qquad
b_{\min}=1.
\]
In this setting,
\[
L_t
\le
\frac12+\beta G^2.
\]
Thus, the curvature is uniformly bounded.

Choosing parameters such as
\[
\beta=1,
\qquad
\eta\in[10^{-4},10^{-3}],
\qquad
\alpha\le10^{-3},
\]
ensures that
\[
F_{0,t},F_{1,t},F_{2,t}
=
O(\eta+\alpha)
\]
remain small relative to typical payoff margins observed in the
simulations.

Consequently, whenever a direction \(v\) satisfies a payoff margin
condition of the form
\[
(\rho(o)-p_t)\cdot v
\ge
\epsilon\|v\|
\]
for some margin \(\epsilon>0\), the conditions
\[
A_t>0,
\qquad
A_t^2>4Q_tF_{0,t},
\]
hold for sufficiently stable updates, yielding a nontrivial interval of
strictly profitable trades.

Therefore, under stable adaptive updates and uniformly interior weights,
LMSR mixtures with arbitrary numbers of outcomes and liquidity levels
admit a nontrivial range of profitable trade sizes whenever the realized
payoff advantage dominates the update-induced distortions.

\section{Experiment Details}\label{app:code}

All experiments were implemented in \texttt{Google Colab}, which provided a convenient environment for rapid iteration and visualization. The simulations use the adaptive mixture market described in the main text with two LMSR liquidity experts, corresponding to a low-liquidity market with $b_1=1$ and a high-liquidity market with $b_2=12$. The mixture temperature is fixed at $\beta=1$ throughout.

To reduce sensitivity to the choice of hybrid-signal coefficients, we extend the base algorithm by introducing a collection of meta-experts. Each meta-expert corresponds to a pair consisting of a liquidity expert and a coefficient profile $m = (a_m,b_m)$ for the hybrid signal. The set of coefficient profiles $\mathcal M$ is defined as
\[\mathcal M = 
\texttt{coeff\_profiles} = \left\{ (6.0, 0.2), \; (3.0, 0.7), \; (1.5, 1.5), \; (0.7, 3.0), \; (0.2, 6.0) \right\}
\]
These profiles range from slippage-dominated objectives to liability-dominated objectives, with an intermediate balanced profile. Thus, the learner is not forced to commit in advance to a single relative weighting of execution cost and inventory risk. Instead, it learns over the product class of liquidity regimes and coefficient profiles.

The learner uses Fixed-Share over this meta-expert class. For each round, the loss assigned to a meta-expert with liquidity index $k$ and coefficient profile $(a_m,b_m)$ is
\[
\ell_{k,m,t}
=
a_m \frac{S_{k,t}-\bar{S}_t}{\sigma_{\text{slip}}}
+
b_m \frac{L_k(q_t)}{\sigma_{\text{liab}}}.
\]
where $\bar{S}_t := \frac{1}{M}\sum_{j=1}^M S_{j,t}$ is the mean slippage across experts. This mean-centering is a standard translation-invariant shift for exponential weights that improves numerical stability without altering the relative theoretical updates.

In the experiments, the normalization constants are set to $\texttt{SLIP\_SCALE}= \sigma_{\text{slip}}=1$ and $ \texttt{LIAB\_SCALE}= \sigma_{\text{liab}}= 4$. 
The scaling constants $\sigma_{\text{slip}}$ and $\sigma_{\text{liab}}$ are used in the experiments to normalize the magnitudes of the two signal components for stable learning, and are omitted from the main text since they do not affect the theoretical formulation and can be absorbed into the coefficients $(a,b)$.
Fixed-Share performs a multiplicative-weights update using learning rate $\eta=5\times 10^{-4}$ and subsequently mixes the updated distribution with the uniform distribution using share parameter $\alpha=10^{-4}$. The share step prevents weights from collapsing completely and allows the learner to recover when the best liquidity regime or coefficient profile changes.

The market itself uses the marginal weights over liquidity experts induced by the meta-expert distribution. That is, after updating the meta-expert weights, we sum the weights of all meta-experts associated with each liquidity level to obtain the liquidity weights used in the mixed cost. The effective liquidity shown in the first panel is computed from the posterior pricing weights rather than directly from the learning weights
\[
b_{\mathrm{eff},t}
=
\left(
\sum_k \frac{\pi_t(k;q_t)}{b_k}
\right)^{-1}.
\]
This provides a curvature-matched liquidity scale for the binary LMSR mixture.

Because the experiments use a binary LMSR market, the state can be reduced to the one-dimensional coordinate
\[
x = q_1-q_2.
\]
By translation invariance of LMSR, all relevant price and cost differences can be represented along this one-dimensional slice, which makes the switch-budget computation tractable. To approximate the positive switching fee, we use a discrete grid $X_{\text{grid}}$ of 5001 points linearly spaced between $-140.0$ and $140.0$, defined as
\[
\texttt{x\_grid}
=
\texttt{np.linspace}(-140.0,140.0,5001).
\]
For each weight update, the fee is computed by evaluating
\[
C_{\mix}(x;w_t)-C_{\mix}(x;w_{t+1})
\]
over this grid, taking the maximum, and then applying the positive part
\[
\mathrm{fee}_t
\approx
\left[
\max_{x\in \texttt{x\_grid}}
\bigl(
C_{\mix}(x;w_t)-C_{\mix}(x;w_{t+1})
\bigr)
\right]_+ .
\]
This numerical approximation is efficient because the binary reduction turns the supremum over states into a one-dimensional search.

The plotted diagnostics report effective liquidity, adaptive slippage, and exposure. Slippage is computed as the Bregman divergence of the relevant cost function. The exposure panel compares the adaptive mixed liability to the liabilities of the fixed low- and high-liquidity LMSR markets. Although the code also tracks the cumulative fee reserve, the displayed protected exposure is computed as the positive part of the mixed liability rather than subtracting the cumulative reserve.



\section{Examples of Adaptive Markets}
\label{app:examples}

We illustrate how the adaptive market framework specializes to two
canonical settings: (i) Arrow--Debreu (categorical) markets and (ii)
pair-betting markets over permutations. In both cases, the adaptive
market is defined by combining a family of cost functions via the
log-sum-exp mixture.

\subsection{Arrow--Debreu markets}

Let $\mathcal O=\{1,\dots,K\}$ with payoff vectors $\rho(i)=e_i$,
so $q \in \mathbb{R}^K$ and prices lie in the simplex $\Delta_K$.

We use scaled LMSR costs
\[
C_{\eta}(q)
=
\eta \log\!\Big(\sum_{i=1}^K e^{q_i/\eta}\Big).
\]

Define experts $C_k := C_{\eta_k}$. The mixture cost is
\[
C_{\mix}(q; w_t)
=
\frac{1}{\beta}
\log\!\Big(
\sum_{k=1}^M w_t(k)\,
\exp\!\big(\beta\, C_{\eta_k}(q)\big)
\Big).
\]

Each expert gradient is
\[
\nabla C_{\eta_k}(q)_i
=
\frac{\exp(q_i/\eta_k)}
{\sum_{j=1}^K \exp(q_j/\eta_k)},
\]
so the mixture gradient becomes
\[
\nabla C_{\mix}(q; w_t)_i
=
\sum_{k=1}^M \pi_t(k;q)\,
\frac{\exp(q_i/\eta_k)}
{\sum_{j=1}^K \exp(q_j/\eta_k)}.
\]
Hence, prices are convex combinations of softmax distributions at
different scales.

\subsection{Pair-betting markets}

Let the finite outcome space be the set of permutations $\mathcal O = S_n$. For each ordered pair
$(i,j)$ with $i \neq j$, define
\[
\rho(\sigma)_{ij} = \mathbf{1}\{\sigma(i) < \sigma(j)\}.
\]
Then $q \in \mathbb{R}^{n(n-1)}$ and
\[
\langle q,\rho(\sigma)\rangle
=
\sum_{i\neq j} q_{ij}\,\mathbf{1}\{\sigma(i)<\sigma(j)\}.
\]

We use the log-partition cost
\[
C(q)
=
\log\!\Big(
\sum_{\sigma \in S_n}
\exp\!\big(\langle q, \rho(\sigma) \rangle\big)
\Big),
\]
and scaled versions $C_\eta(q)=\eta C(q/\eta)$.

Define experts $C_k := C_{\eta_k}$. Then the mixture cost is
\[
C_{\mix}(q; w_t)
=
\frac{1}{\beta}
\log\!\Big(
\sum_{k=1}^M w_t(k)\,
\exp\!\big(\beta\, C_{\eta_k}(q)\big)
\Big).
\]

Each expert gradient is
\[
\nabla C_{\eta_k}(q)_{ij}
=
\mathbb{P}_{q/\eta_k}\!\big(\sigma(i)<\sigma(j)\big),
\]
where
\[
\mathbb{P}_{q/\eta_k}(\sigma)
=
\frac{\exp(\langle q/\eta_k,\rho(\sigma)\rangle)}
{\sum_{\tau\in S_n}\exp(\langle q/\eta_k,\rho(\tau)\rangle)}.
\]

Thus, the mixture gradient is
\[
\nabla C_{\mix}(q; w_t)_{ij}
=
\sum_{k=1}^M \pi_t(k;q)\,
\mathbb{P}_{q/\eta_k}\!\big(\sigma(i)<\sigma(j)\big),
\]
i.e., a convex combination of pairwise probabilities induced by each
expert.



\end{document}